\documentclass[12pt,one column]{IEEEtran}

\usepackage{setspace}
\pdfoutput=1
\usepackage{amsmath}
\usepackage{amssymb}
\usepackage{cite}
\usepackage{color}
\usepackage{epsfig}
\usepackage{epsf}
\usepackage{rotating}
\usepackage{mathrsfs}
\usepackage{epsfig}
\usepackage{graphics}
\usepackage{theorem}

\usepackage{amsmath}
\usepackage{amssymb}
\usepackage{cite}
\usepackage{color}
\usepackage{epsfig}
\usepackage{epsf}
\usepackage{rotating}
\usepackage{mathrsfs}
\usepackage{epsfig}
\usepackage{graphics}
\usepackage{bbold}
\usepackage{theorem}

\newtheorem{thm}{Theorem}
\newtheorem{lem}{Lemma}

\newtheorem{corollary}{Corollary}
\newtheorem{proposition}{Proposition}

\newtheorem{definition}{Definition}

\begin{document}

\title{{  Randomized vs. Orthogonal Spectrum Allocation in Decentralized  Networks:
Outage Analysis}}
\author{ Kamyar Moshksar, Alireza Bayesteh and Amir K. Khandani \\
\small Coding \& Signal Transmission Laboratory (www.cst.uwaterloo.ca)\\
Dept. of Elec. and Comp. Eng., University of Waterloo\\ Waterloo, ON, Canada, N2L 3G1 \\
Tel: 519-725-7338, Fax: 519-888-4338\\e-mail: \{kmoshksa, alireza,
khandani\}@cst.uwaterloo.ca} \maketitle
\begin{abstract}
We address a decentralized wireless communication network with a fixed number $u$ of
frequency sub-bands to be shared among $N$
transmitter-receiver pairs. It is assumed that the number of users $N$ is a random variable with a given distribution and the channel gains are quasi-static Rayleigh fading. The transmitters are assumed to be unaware of the number of active users in the network as well as the channel gains and not capable of detecting  the presence of other users in a given frequency sub-band. Moreover, the users are unaware of each other's codebooks and hence, no multiuser detection is possible.  We consider a randomized Frequency Hopping (FH) scheme in which each transmitter randomly hops over a subset of the $u$ sub-bands from transmission to transmission.  Developing a new upper bound on the differential entropy of a mixed Gaussian random vector and using entropy power inequality, we offer a series of lower bounds on the achievable rate of each user. Thereafter, we obtain lower bounds on the maximum transmission rate per user to ensure a specified outage probability at a given Signal-to-Noise Ratio (SNR) level. We demonstrate that  the so-called outage capacity can be considerably higher in the FH scheme than in the Frequency Division (FD) scenario for reasonable distributions on the number of active users. This guarantees a higher spectral efficiency in FH compared to FD. \end{abstract}
\vskip 1.5cm

\section{Introduction}
\subsection{Motivation and Related Works}
Optimal resource allocation is an imperative issue in wireless
networks. Wide applications of  wireless systems, in recent years, and the limited available resources in the network necessitate efficient usage of such resources.  Multiuser interference is known to be the most important factor, which degrades the network performance when multiple users share the same spectrum. There has been a tremendous amount of research on designing an efficient and low complexity resource allocation scheme that maximizes the quality of service per user while controlling the detrimental effect of multi-user interference. Existing resource allocation schemes in the literature are classified as \textit{centralized}, i.e., a central controller manages the resources, or \textit{decentralized}, where resource allocation is performed locally at each node.

 The main goal in traditional wireless systems was to avoid the interference between users by transmitting over orthogonal channels. A well-known example of such systems is the Frequency Division (FD) system, in which different users
transmit over disjoint frequency sub-bands. The assignment of frequency sub-bands is performed by a central controller. Despite its simplicity, \cite{1} proves that in a wireless network in which the interference is treated as noise (no multi-user detection is performed),  and if the
crossover gains are sufficiently greater than the forward gains, FD is Pareto-optimal. Due to practical
considerations, such FD systems have a fixed infrastructure, i.e., they rely on a fixed number of
 frequency sub-bands. If the number of users
changes, the system is not guaranteed to offer the best possible
spectral efficiency because, most of the time, the majority of the potential users may be
inactive.  In recent
years, many centralized power and spectrum allocation schemes have been extensively
studied in cellular and multihop wireless networks
\cite{ KumaranITWC0505, ElBattITWC0104,HollidayISIT2004,WassermanITWC0603, HanITC0805, KatzelaIPC0696, KianiISIT2006,
LiangITIT1007}. Clearly,
centralized resource allocation schemes provide a significant
improvement in the network throughput over decentralized
(distributed) approaches. However, they require extensive knowledge
of the network configuration. In particular, when the number of
nodes is large, deploying such centralized schemes may not be
practically feasible.

 Most of the decentralized schemes in the literature rely on either \textit{game-theoretic} approaches or \textit{cognitive radios}.  Cognitive radios \cite{2} have the ability to sense the unoccupied
portion of the available spectrum. Fundamental limits of wireless
networks with cognitive radios are studied in \cite{3,4,5,6,7}.
These smart radios require sophisticated techniques do detect the spectrum holes that
add to the overall system complexity \cite{8}. As such, it is essential to have a decentralized spectrum sharing strategy without using cognitive radios, which allows the users to coexist while utilizing the spectrum efficiently and fairly.

Due to its interference avoidance nature, hopping is the simplest
spectrum sharing method to use in decentralized networks. Frequency
Hopping (FH) is the most popular scenario in this category in which
users randomly switch to different frequency sub-bands from
transmission to transmission. As different users typically have no
prior information about the codebooks of the other users, the most
efficient method is avoiding interference by choosing unused
channels. As mentioned earlier, searching the spectrum to find
spectrum holes is not an easy task due to the dynamic spectrum
usage. As such, FH  is a realization of transmission
without sensing, while avoiding the collisions as much as possible.
Frequency Hopping is one of the standard signaling schemes\cite{15}
adopted in ad hoc networks. In short range scenarios, bluetooth
systems \cite{19,20,21} are the most popular examples of a wireless
personal area network, or WPAN.  Using FH over the
unlicensed ISM band, a bluetooth system provides robust
communication to unpredictable sources of interference. A
modification of Frequency Hopping called Dynamic Frequency Hopping
(DFH) selects the hopping pattern based on interference measurements
in order to avoid dominant interferers. The performance of a DFH
scheme when applied to a cellular system is assessed in
\cite{22,23,24}.

Although there has been  a tremendous amount of work on the performance
evaluation of hopping-based decentralized networks, there are only a
few information-theoretic results reflecting the fundamental limits
of such networks. In a pioneering work, \cite{kamyar1} offers a
clean analysis of a decentralized network where all users follow a
randomized hopping strategy to share the resources, while the number
of active users is a random variable with a given distribution. In
\cite{kamyar1}, the channel gains and the number of active users are
assumed to be static and known to the corresponding
transmitter-receiver nodes. This assumption makes the concept of
achievable rate in the Shannon sense meaningful.
 However, in case the channel gains and the number of active users are not known to the transmitters, the concept of achievable rate may no longer be valid. A common setup for such an assumption is a network where channel gains are quasi-static fading and unknown to the transmitters, which is the framework for this paper.

Rayleigh fading is an unavoidable phenomenon in wireless networks
that can affect the performance of the system significantly.
Traditionally, Rayleigh fading has been considered to be harmful due
to reducing the transmission reliability in wireless networks.
However, recently, researchers have been able to reduce this harmful
effect by exploiting the so-called \textit{multiuser diversity}
\cite{knopp,tse}. This can be considered a \textit{scheduling gain}
by allowing the users with favorable channels to be active. It is
shown that multiuser diversity gain can be as large as $\log \log N$
in broadcast and multiple-access channels \cite{sharif,alireza,yoo}
and $\log N$ in single-hop ad hoc networks \cite{masoud1,jamshid1}, as
$N$ grows to infinity. However, achieving this scheduling gain
requires the fading channels to vary over time such that all
possible realizations of the fading process are covered. In the case
that channel gains are selected randomly at the start of the
transmission and remain constant during the whole transmission
period (quasi-static fading), the channels do not have ergodic
behavior. In this case, a suitable performance measure is the
$\epsilon$-\textit{outage capacity} \cite{shamai}, denoted by
$R(\epsilon)$, which is defined as the  maximum transmission rate
per user,  ensuring an outage probability below $\epsilon$, i.e.,
$$R(\epsilon)=\sup\{R: \Pr\{\mathrm{Outage}\}<\epsilon\}.$$

The reality of wireless channel is more complicated to be simply
represented by Rayleigh fading model. A class of channel models considered in the literature is the one in
which the signal power decays according to a distance-based
attenuation law \cite{M1,M2,M3,M4,M5,M6,M7,M8}. Moreover,  the presence of obstacles
adds some randomness (known as shadowing) to the received signal. It
is well known that the effects of such random phenomena can
significantly affect the throughput of a spectrum sharing network in
both multi-hop \cite{M9,M10,M11,M12} and single-hop scenarios \cite{M13}
(Chapter 8), \cite{M14,M15,M16,M17,M18}. These features indeed increase the frequency reuse factor as they
will attenuate the interference caused by a given transmitter on its
neighboring receivers.
In spite of the significance of the
effects of distance-based attenuation and shadowing on the
throughput of a spectrum sharing system, unfortunately, there is not
a single commonly accepted model to represent these factors. It
should be emphasized that the inclusion of signal attenuation due to
distance and/or shadowing will indeed simplify the spectrum sharing
as the multi-user interference will be attenuated and consequently
its harmful effect will be reduced. On the other hand, these factors
do not impact the performance of the orthogonal schemes in which the
multi-user interference is altogether avoided. In spite of this
fact, as the actual model used to represent distance-based
attenuation and shadowing can have a profound impact on the system
throughput (to the advantage of the randomized spectrum sharing
schemes advocated in the current article), to avid any confusion, we
have relied on a simple Rayleigh fading model which in some sense
captures the minimum advantage offered by the proposed scheme vs.
those based on orthogonal separation of users. 

Reference
\cite{masoud} studies a wireless network composed of a set of
transmitter/receiver pairs in which a given link can be off or
transmit with a constant power. \cite{masoud} considers both the
case of  Rayleigh fading as well as a Rayleigh fading mixed with a
proper distance-based attenuation and among other results provides a
comparison between the scaling (with respect to the number of links)
of the throughput in these two cases.

In \cite{Jindal,Jindall}, the authors study  a decentralized
wireless ad hoc network where different transmitters are connected
to different receivers through channels with a similar path loss
exponent. Assuming the transmitters are scattered over the two
dimensional plane according to a Poisson point process, a fixed
bandwidth is partitioned into a certain number of sub-bands, such
that the so-called transmission intensity in the network is
maximized, while the probability of outage per user is below a
certain threshold\cite{Jindal}. The transmission strategy is based
on choosing one sub-band randomly per transmission, which is a
special case of Frequency Hopping. In \cite{Jindall}, a
non-iterative and distributed power control scheme is introduced, for
which the constant power and the channel inversion schemes are
extreme cases. It is observed that none of these cases are ideal in
general. In fact, it is shown that regulating the transmission power
proportionately to the inverse square root of the forward channel
strength minimizes the outage probability.

Recently, Orthogonal Frequency Division Multiplexing (OFDM) has been
considered as a promising technique in many wireless technologies.
OFDM partitions a wide-band channel to a group of narrow-band
orthogonal sub-channels. This motivates us to consider the
underlying system to consist of $u$ narrow-band orthogonal frequency
sub-bands.

\subsection{Our Contribution}
 In this paper, we consider a decentralized wireless communication network with a fixed number $u$ of
frequency sub-bands to be shared among $N$
transmitter-receiver pairs. Any transmitter is connected to any receiver through a channel with quasi-static and non-frequency selective Rayleigh fading. In other words, the channel gains are picked randomly (based on Rayleigh distribution) at the start of the transmission and remain fixed for the whole transmission. It is assumed that the number of active users is a random variable with a given probability mass function. The channel gains and the number of active users are unknown to all transmitters, however, the receivers are assumed to be aware of their direct channel gains and the interference Probability Density Function (PDF). Moreover, users are unaware of each other's codebooks and hence, no multiuser detection is possible. A randomized Frequency Hopping scheme is proposed in which each transmitter randomly hops over $v$ out of $u$ sub-bands from transmission to transmission. Assuming i.i.d. Gaussian signals are transmitted over the chosen sub-bands, the distribution of the noise plus interference becomes mixed Gaussian, which makes calculation of the achievable rate complicated.
The main contributions of the paper are:

\begin{itemize}
 \item   Developing a new upper bound on the differential entropy of a class of mixed Gaussian random vectors and using entropy power inequality, we offer three lower bounds on the $\epsilon$-outage capacity of each user denoted by $R_{\mathrm{FH,lb}}^{(1)} (\epsilon)$,  $R_{\mathrm{FH,lb}}^{(2)} (\epsilon)$, and $R_{\mathrm{FH,lb}}^{(3)} (\epsilon)$. To evaluate the system performance analytically, we use $R_{\mathrm{FH,lb}}^{(3)}(\epsilon)$, which can be computed easily. However, computation of  $R_{\mathrm{FH,lb}}^{(1)} (\epsilon)$ and  $R_{\mathrm{FH,lb}}^{(2)} (\epsilon)$ involves integrations that cannot be carried out in a closed form. In the simulation results, we use the lower bounds $R_{\mathrm{FH,lb}}^{(1)} (\epsilon)$ and  $R_{\mathrm{FH,lb}}^{(2)} (\epsilon)$, which are tighter than $R_{\mathrm{FH,lb}}^{(3)} (\epsilon)$.

 \item We perform asymptotic analysis for the outage capacity in terms of $\epsilon$ and SNR. In the asymptotically small $\epsilon$ regime, we observe that the maximum of outage capacity is obtained  for either $v=1$ or $v=u$. In the asymptotically small SNR regime, we demonstrate that for any value of $v$ the system achieves the optimal performance. For asymptotically large values of SNR, it is shown that $v_{\mathrm{opt}}=\left\lceil \frac{u}{n_{\max}}\right\rceil$, where $n_{\max}$ is the maximum possible number of concurrently active users in the network.

 \item We compare the outage capacity of the underlying FH scenario with that of the FD scheme for various scenarios in terms of  distributions on the number of active users, SNR and $\epsilon$. It is shown that FH outperforms FD in terms of outage capacity in many cases. We observe that in the low SNR regime, FH and FD offer the same performance. In the low $\epsilon$ regime, FD is always better than FH, however, for many practical scenarios there exists a threshold, $\epsilon_{\mathrm{th}}$, such that FH outperforms FD as far as $\epsilon \geq \epsilon_{\mathrm{th}}$. Also, we have shown that supremacy of FH over FD in the high SNR regime occurs quite often.
 \end{itemize}

 The paper outline is as follows. The system model is given in section II. Section III describes analysis of the outage capacity. Section IV is devoted to derive lower bounds on the achievable rates of users. Also, in this section, we offer a new computable upper bound on the differential entropy of a mixed Gaussian random vector. In section V, based on the results in sections III and IV, we discuss how the users in the FH system fairly share the spectrum while maximizing the outage capacity. Derivation of various lower bounds on the outage capacity of users is part of the materials in this section. Comparison between the FH and FD scenarios is given in section VI through the analysis and simulation results. Finally, section VII concludes the paper.

 \subsection{Notation}
Throughout the paper, we use the notation $\mathrm{E} \{.\}$ for the
expectation operator. For a function $f(X,Y)$ of two independent random variables $X$ and $Y$, $\mathrm{E}_{X}\left\{f(X,Y)\right\}$ denotes the expectation of $f(X,Y)$ with respect to $X$ while $Y$ is treated as a parameter. $\mathrm{Pr}\{\mathcal{E}\}$ denotes the
probability of an event $\mathcal{E}$, $\mathbb{1}(\mathcal{E})$ the
indicator function of an event $\mathcal{E}$ and $p_{X}(.)$  the PDF of a random variable $X$. Also,
$\mathrm{I} (X;Y)$ denotes the mutual information between random
variables $X$ and $Y$ and $\mathrm{h} (X)$ denotes the differential entropy of
a continuous random variable $X$.

\section{System Model and Assumptions}
We consider a wireless network with $N$ users\footnote{Each user consists of a transmitter-receiver pair.} operating on a bandwidth consisting of $u$ sub-bands. It is assumed that the $i^{th}$ user exploits $v_{i}$  out of the $u$  sub-bands in each transmission and hops randomly to another set of $v_{i}$ frequency sub-bands in the next transmission.  This user transmits independent complex Gaussian signals\footnote{Since in this work we deal with fading channels, the transmitted signals are assumed to be complex for simplicity of analysis.} of variance $\frac{P}{v_{i}}$ over each of the chosen sub-bands in which $P$ denotes the total average power of each transmitter. Each receiver is assumed to know the hopping pattern of its affiliated transmitter. It is assumed that the users are not aware of each other's codebooks and hence, no interference cancellation is performed at the receiver sides. The quasi-static and non frequency-selective fading coefficient\footnote{The knowledge of the channel gains and the number of active users at the receiver side is realized by identifying the interference PDF.} of the channel from the $i^{th}$ transmitter to the $j^{th}$ receiver is shown by $h_{i,j}$. All the channel coefficients in the network are assumed to be complex zero-mean Gaussian random variables of unit variance corresponding to Rayleigh fading. However, due to the absence of any feedback link, the transmitters do not have information about any of the channel gains. By the same token, the transmitters are not aware of the number of active users in the network.

As all users hop over different portions of the spectrum from transmission to transmission, no receiver is assumed to be capable of tracking the instantaneous interference. This assumption makes the interference plus noise PDF at the receiver side of each user be mixed Gaussian. In fact, depending on different choices the other users make to select the frequency sub-bands and values of the crossover gains, the interference on each frequency sub-band at the receiver side of any user has up to $2^{N-1}$ power levels\footnote{It is notable that the interference plus noise PDF has $2^{N-1}$ power levels, almost surely, as the channel gains are considered to be continuous random variables.}. The vector consisting of the received signals on the frequency sub-bands at the $i^{th}$ receiver in a typical transmission slot is
\begin{equation}
\vec{Y}_{i}=h_{i,i}\vec{X}_{i}+\vec{Z}_{i},
\end{equation}
where $\vec{X}_{i}$ is the $u\times 1$ transmitted vector and $\vec{Z_{i}}$ is the noise plus interference vector on the receiver side of the $i^{th}$ user. Due to the fact that each transmitter hops randomly from transmission to transmission, one may write $p_{\vec{X}_{i}}(.)$ as
\begin{eqnarray}
p_{\vec{X}_{i}}(\vec{x})=\sum_{C\in \mathfrak{C}}\frac{1}{{u\choose v_{i}}}g_{u}(\vec{x},C),
\end{eqnarray}
 which corresponds to the mixed Gaussian distribution. In the above equation, $g_{u}(\vec{x},C)$ denotes the PDF of a $u\times 1$ complex zero-mean jointly Gaussian vector of covariance matrix $C$ and the set $\mathfrak{C}$ includes all $u\times u$ diagonal matrices in which $v_{i}$ out of the $u$ diagonal elements are $\frac{P}{v_{i}}$ and the rest are zeros. Denoting the noise plus interference on the $j^{th}$ sub-band at the receiver side of the $i^{th}$ user by $Z_{i,j}$ (the $j^{th}$ component of $\vec{Z}_{i}$), it is clear that $p_{Z_{i,j}}(.)$ is not dependent on $j$. This is by the fact that crossover gains are not sensitive to frequency and there is no particular interest to a specific frequency sub-band by any user. We assume there are $L_{i}+1$ ($L_{i}\leq 2^{N-1}-1$) possible nonzero power levels for $Z_{i,j}$, say $\{\sigma^{2}_{i,l}\}_{l=0}^{L_{i}}$. Denoting the occurrence probability of $\sigma^{2}_{i,l}$ by $a_{i,l}$, $p_{Z_{i,j}}(.)$ is given by
\begin{equation}
\label{hf}
p_{Z_{i,j}}(z)=\sum_{l=0}^{L_{i}}\frac{a_{i,l}}{\pi\sigma_{i,l}^{2}}\exp \left(-\frac{|z|^2}{\sigma_{i,l}^{2}} \right),
\end{equation}
where  $\sigma^{2}=\sigma_{i,0}^{2}< \sigma_{i,1}^{2}< \sigma_{i,2}^{2}<...<\sigma_{i,L_{i}}^{2}$ ($\sigma^{2}$ is the ambient noise power). We notice that for each $l\geq 0$, there exists a $c_{i,l} \geq 0$ such that $\sigma_{i,l}^{2}=\sigma^{2}+c_{i,l}P$ where $0=c_{i,0}<c_{i,1}<c_{i,2}<...<c_{i,L_{i}}$. In fact, one may write $Z_{i,j}=  \sum_{\substack{k=1\\k\neq i}}^{N}\epsilon_{k,j}h_{k,i}X_{k,j}+\nu_{i,j}$ where $X_{k,j}$ is the signal of the $k^{th}$ user sent on the $j^{th}$ sub-band, $ \epsilon_{k,j}$ is a Bernoulli random variable showing if the $k^{th}$ user has utilized the $j^{th}$ sub-band and $\nu_{i,j}$ is the ambient noise, which is a zero-mean complex Gaussian random variable with variance $\sigma^{2}$. The ratio $\frac{P}{\sigma^{2}}$ is taken as a measure of SNR as is denoted by $\gamma$ throughout the paper.

Since the transmitters are not aware of the channel gains and the number of active users in the network, the Shannon capacity is not meaningful in this setup. In this case, a suitable performance measure is the $\epsilon$-\textit{outage capacity}, denoted by $R(\epsilon)$, which is defined as the  maximum transmission rate per user  ensuring an outage probability below $\epsilon$, i.e.,
\begin{equation}R(\epsilon)=\sup\{R: \Pr\{\mathrm{Outage}\}<\epsilon\}.\end{equation}

\section{Analysis of the Outage Capacity}
Let $\vec{h}_{i}$  contain the channel coefficients concerning the $i^{th}$ user, i.e.,  $\vec{h}_{i}=\begin{pmatrix}
    h_{1,i}  & \cdots& h_{N,i}
\end{pmatrix}^{T}$. In this case, we denote the achievable rate of the $i^{th}$ user by  $\mathscr{R}_{i}(\vec{h}_{i})$. The outage event for this user is
\begin{equation}
 \mathcal{O}_{i} (R)\triangleq\{\vec{h}_{i}: \mathscr{R}_{i} (\vec{h}_{i}) <R\},
\end{equation}
where $R$ is the actual transmission rate. Hence,
\begin{equation}
\label{kootool}
 R(\epsilon)=\mathrm{sup}\Big\{R: \Pr\{\mathcal{O}_{i} (R)\}<\epsilon\Big\}.
\end{equation}
We emphasize that the randomness of the number of active users is involved in the outage event, as $N$ represents the size of $\vec{h}_{i}$. Moreover, due to symmetry, the $\epsilon$-outage capacity is the same for all users.

Having $\vec{h}_{i}$ fixed, it can be observed that the communication channel of the $i^{th}$ user is a channel with state $S_{i}$, the hopping pattern of the $i^{th}$ user, which is independently changing over different transmissions and known to both the transmitter and receiver. Hence,
\begin{equation}
\mathscr{R}_{i} (\vec{h}_{i})=\mathrm{I}(\vec{X}_{i};\vec{Y}_{i}\vert S_{i})=\sum_{s_{i}\in \mathfrak{S}_{i}}\Pr(S_{i}=s_{i})\mathrm{I}(\vec{X}_{i};\vec{Y}_{i}\vert S_{i}=s_{i}),
\end{equation}
where $\mathrm{I}(\vec{X}_{i};\vec{Y}_{i}\vert S_{i}=s_{i})$ is the mutual information between $\vec{X_{i}}$ and $\vec{Y_{i}}$ for the specific sub-band selection dictated by $S_{i}=s_{i}$. The set $\mathfrak{S}_{i}$ denotes all possible selections of $v_{i}$ out of the $u$ sub-bands. As $p_{\vec{Z_{i}}}(.)$ is a symmetric density function, meaning all its components have the same PDF given in (\ref{hf}), we deduce that $\mathrm{I}(\vec{X}_{i};\vec{Y}_{i}\vert S_{i}=s_{i})$ is independent of $s_{i}$. Therefore, to calculate $\mathscr{R}_{i}(\vec{h}_{i})$, we may assume any specific sub-band selection for the $i^{th}$ user in $\mathfrak{S}_{i}$, say the first $v_{i}$ sub-bands. Denoting this specific state by $s^{*}_{i}$, we get
\begin{equation}
\mathscr{R}_{i} (\vec{h}_{i})=\mathrm{I}(\vec{X}_{i};\vec{Y}_{i}\vert S_{i}=s^{*}_{i}).
\end{equation}
In this case, we denote $\vec{Y}_{i}$ and $\vec{X}_{i}$ by $\vec{Y}_{i}(s^{*}_{i})$ and $\vec{X}_{i}(s^{*}_{i})$ respectively. Obviously,
\begin{equation}
\label{lala}
\mathscr{R}_{i} (\vec{h}_{i})=\mathrm{I}(\vec{X}_{i}(s^{*}_{i});\vec{Y}_{i}(s^{*}_{i}))=\mathrm{h}(\vec{Y}_{i}(s^{*}_{i}))-\mathrm{h}(\vec{Z}_{i}).
\end{equation}
As $\vec{Y}_{i}(s^{*}_{i})$ and $\vec{Z}_{i}$ are complex mixed Gaussian vectors, there is no closed expression for the differential entropy of these vectors. As such, we provide a lower bound $\mathscr{R}_{i,\mathrm{lb}}(\vec{h}_{i})$ on $\mathscr{R}_{i}(\vec{h}_{i})$ in the following section. Subsequently, using $\mathscr{R}_{i,\mathrm{lb}} (\vec{h}_{i})$, we derive a lower bound on the outage capacity of the $i^{th}$ user as
\begin{equation}
 R_{\mathrm{lb}}(\epsilon)=\mathrm{sup}\Big\{R: \Pr\{\vec{h}_{i}:\mathscr{R}_{i,\mathrm{lb}}(\vec{h}_{i})<R\}<\epsilon\Big\},
\end{equation}
and show that this lower bound is higher than the actual outage capacity in the FD scheme in many scenarios.
\section{Lower Bounds on $\mathscr{R}_{i}(\vec{h}_{i})$}

The aim of this section is to find a lower bound on $\mathscr{R}_{i} (\vec{h}_{i})$.  The idea behind deriving this lower bound is to invoke entropy power inequality (EPI). As we will see, this initial lower bound is not in a closed form as it depends on the differential entropy of a mixed Gaussian random variable. In appendix A,  we obtain an appropriate upper bound on such an entropy, which leads us to the final lower bound on $\mathscr{R}_{i} (\vec{h}_{i})$.

\begin{thm} \label{kuyyy1}
There exists a lower bound on $\mathscr{R}_{i} (\vec{h}_{i})$ given by
\begin{eqnarray}
 \mathscr{R}_{i} (\vec{h}_{i}) &\geq& v_{i}\log\left(\frac{2^{-\mathscr{H}_{i}}2^{\mathscr{G}_{i}}\vert h_{i,i}\vert^{2}\gamma}{v_{i}\prod_{l=1}^{L_{i}}(c_{i,l}\gamma+1)^{a_{i,l}}}+1\right),
\end{eqnarray}
where
$\mathscr{H}_{i}=-\sum_{l=0}^{L_{i}}a_{i,l}\log a_{i,l}$ and
$\mathscr{G}_{i}=\frac{\sigma}{\sigma_{i,L_{i}}}\sum_{l=1}^{L_{i}}a_{i,l}\log \left(1+\frac{\sigma_{i,L_{i}}}{\sigma}\frac{\sum_{m=0}^{l-1}a_{i,m}}{a_{i,l}} \right)$.
\end{thm}

\begin{proof}
Let us define $\vec{X}'_{i}$ to be the $v_{i}\times 1$ signal vector\footnote{$\vec{X}'_{i}$ consists of the first $v_{i}$ elements of $\vec{X}_{i}(s_{i}^{*})$.} of the $i^{th}$ transmitter that is sent through the first $v_{i}$ frequency sub-bands. Let $\vec{Y}'_{i}=h_{i,i}\vec{X}'_{i}+\vec{Z}'_{i}$, where $\vec{Z}'_{i}$ is the noise plus interference vector at the receiver side of the $i^{th}$ user on the first $v_{i}$ frequency sub-bands. According to the classical EPI\footnote{As we deal with complex random vectors, the format of EPI is different from its real counterpart.},
\begin{equation}
2^{\frac{1}{v_{i}}\mathrm{h}(\vec{Y}'_{i})}\geq 2^{\frac{1}{v_{i}}\mathrm{h}(h_{i,i}\vec{X}'_{i})}+2^{\frac{1}{v_{i}}\mathrm{h}(\vec{Z}'_{i})}.
\end{equation}
Dividing both sides by $2^{\frac{1}{v_{i}}\mathrm{h}(\vec{Z}'_{i})}$,
\begin{equation}
\label{piu}
\mathrm{h}(\vec{Y}'_{i})-\mathrm{h}(\vec{Z}'_{i})\geq v_{i}\log\left(2^{\frac{1}{v_{i}}\left(\mathrm{h}(h_{i,i}\vec{X}'_{i})-\mathrm{h}(\vec{Z}'_{i})\right)}+1\right).
\end{equation}
On the other hand, since $\vec{Y}'_{i}$ is a subvector of $\vec{Y}_{i}(s^{*}_{i})$,
\begin{equation}
\label{qiu}
\mathscr{R}_{i} (\vec{h}_{i}) =\mathrm{I}(\vec{X}_{i}(s^{*}_{i});\vec{Y}_{i}(s^{*}_{i}))\geq \mathrm{I}(\vec{X}'_{i};\vec{Y}'_{i})=\mathrm{h}(\vec{Y}'_{i})-\mathrm{h}(\vec{Z}'_{i}).
\end{equation}
Based on (\ref{piu}) and (\ref{qiu}), we get the following lower bound on $\mathscr{R}_{i}(\vec{h}_{i})$,
\begin{equation}
\label{uuu}
\mathscr{R}_{i} (\vec{h}_{i}) \geq v_{i}\log(2^{\frac{1}{v_{i}}\left(\mathrm{h}(h_{i,i}\vec{X}'_{i})-\mathrm{h}(\vec{Z}'_{i})\right)}+1).
\end{equation}
Clearly, $\mathrm{h}(h_{i,i}\vec{X}'_{i})=v_{i}\log\left(\pi e\frac{\vert h_{i,i}\vert^{2} P}{v_{i}}\right)$. As $\vec{Z}'_{i}$ has a mixed Gaussian distribution, there is no closed formula for $\mathrm{h}(\vec{Z}'_{i})$. To circumvent this difficulty, we have to find an appropriate upper bound on $\mathrm{h}(\vec{Z}'_{i})$.

We start with the following Lemma:
  \begin{lem} \label{dool1}
  Let $\vec{\Theta}$ be a $t\times 1$ complex mixed Gaussian random vector with different covariance matrices $\{\varrho_{l}^{2}I_{t}\}_{l=1}^{L}$ and corresponding probabilities $\{p_{l}\}_{l=1}^{L}$ where $\varrho_{1}^{2}<\cdots<\varrho_{L}^{2}$ and $\sum_{l=1}^{L}p_{l}=1$. Then,
  \begin{equation}
 \mathrm{h}(\vec{\Theta})\leq t \sum_{l=1}^{L}p_{l}\log \left(\pi e\varrho_{l}^{2} \right)+\mathscr{H}-\mathscr{G}
  \end{equation}
  where \begin{equation}\mathscr{H}=-\sum_{l=1}^{L} p_{l}\log p_{l}\end{equation} and
  \begin{equation}
  \mathscr{G}=\frac{\varrho_{1}^{2t}}{\varrho_{L}^{2t}}\sum_{l=2}^{L}p_{l}\log \left(1+\frac{\varrho_{L}^{2t}}{\varrho_{1}^{2t}}\frac{\sum_{m=1}^{l-1}p_{m}}{p_{l}}\right).\end{equation}

    \end{lem}

  \begin{proof}
  See appendix A.
  \end{proof}

Using the chain rule for differential entropy,
\begin{equation}
\label{hjk}
\mathrm{h}(\vec{Z}'_{i})\leq\sum_{j=1}^{v_{i}}\mathrm{h}(Z_{i,j}).
\end{equation}
Applying the special case of Lemma \ref{dool1} corresponding to $t=1$ for a scalar complex mixed Gaussian random variable given in (\ref{hf}),
\begin{equation}
\label{ma}
\mathrm{h}(Z_{i,j})\leq \sum_{l=0}^{L_{i}}a_{i,l}\log \left(\pi e\sigma_{i,l}^{2}\right)+\mathscr{H}_{i}-\mathscr{G}_{i},\end{equation}
where
\begin{equation}\mathscr{H}_{i}=-\sum_{l=0}^{L_{i}}a_{i,l}\log a_{i,l}
\end{equation}
and
\begin{eqnarray} \label{ol}
\mathscr{G}_{i} &=&\frac{\sigma^2}{\sigma^2_{i,L_{i}}}\sum_{l=1}^{L_{i}}a_{i,l}\log \left(1+\frac{\sigma^2_{i,L_{i}}}{\sigma^2}\frac{\sum_{m=0}^{l-1}a_{i,m}}{a_{i,l}} \right)\notag\\
&=& \frac{1}{c_{i,L_i} \gamma +1} \sum_{l=1}^{L_{i}}a_{i,l}\log \left(1+\frac{(c_{i,L_i} \gamma +1) \sum_{m=0}^{l-1}a_{i,m}}{a_{i,l}} \right).
\end{eqnarray}
Substituting (\ref{ma}) in (\ref{hjk}) gives an upper bound on $\mathrm{h}(\vec{Z}'_{i})$ as
\begin{eqnarray} \label{gooz1}
 \mathrm{h}(\vec{Z}'_{i})\leq v_i\sum_{l=0}^{L_{i}}a_{i,l}\log \left(\pi e\sigma_{i,l}^{2}\right)+v_i \left(\mathscr{H}_{i}-\mathscr{G}_{i} \right).
\end{eqnarray}
By (\ref{gooz1}) and (\ref{uuu}),
\begin{eqnarray} \label{kood}
\mathscr{R}_{i} (\vec{h}_i)&\geq& v_{i}\log\left(2^{\frac{1}{v_{i}}\left(\log\left(\pi e\frac{\vert h_{i,i}\vert^{2}P}{v_{i}}\right)^{v_{i}}-v_{i}\big(\sum_{l=0}^{L_{i}}a_{i,l}\log(\pi e\sigma_{i,l}^{2})+\mathscr{H}_{i}-\mathscr{G}_{i}\big)\right)}+1\right) \notag\\
\label{uuku}
&=& v_{i}\log\left(\frac{2^{-\mathscr{H}_{i}}2^{\mathscr{G}_{i}}\vert h_{i,i}\vert^{2}\gamma}{v_{i}\prod_{l=1}^{L_{i}}(c_{i,l}\gamma+1)^{a_{i,l}}}+1\right).
\end{eqnarray}
\end{proof}
Let us define
\begin{equation}
\label{ }
\mathscr{R}^{(1)}_{i,\mathrm{lb}} (\vec{h}_i)\triangleq  v_{i}\log\left(\frac{2^{-\mathscr{H}_{i}}2^{\mathscr{G}_{i}}\vert h_{i,i}\vert^{2}\gamma}{v_{i}\prod_{l=1}^{L_{i}}(c_{i,l}\gamma+1)^{a_{i,l}}}+1\right).\end{equation}

The above proof reveals the following observations:

\textit{Observation 1-} As $L_{i}=2^{N-1}-1$ with a probability of $1$, it can be immediately verified that $\mathscr{H}_{i}$ does not depend on the crossover gains. However, $\mathscr{G}_{i}$  is implicitly a function of all crossover gains as the partial sums $\sum_{m=1}^{l-1}a_{i,m}$ for $2\leq l\leq L_{i}$ depend on the ordering of the crossover gains. This will be  investigated more in Lemma \ref{dool2}.


\textit{Observation 2-} Since $\prod_{l=1}^{L_{i}}(c_{i,l}\gamma+1)^{a_{i,l}}\leq \prod_{l=1}^{L_{i}}(c_{i,L_{i}}\gamma+1)^{a_{i,l}}=(c_{i,L_{i}}\gamma+1)^{(1-a_{i,0})}$, one obtains a looser version of $\mathscr{R}^{(1)}_{i,\mathrm{lb}} (\vec{h}_i)$ given by
\begin{equation}
\label{koodie}
\mathscr{R}^{(2)}_{i,\mathrm{lb}} (\vec{h}_i)\triangleq v_{i}\log\left(\frac{2^{-\mathscr{H}_{i}}2^{\mathscr{G}_{i}}\vert h_{i,i}\vert^{2}\gamma}{v_{i}(c_{i,L_{i}}\gamma+1)^{1-a_{i,0}}}+1\right).
\end{equation}
We note that $\mathscr{R}^{(2)}_{i,\mathrm{lb}} (\vec{h}_i)$ still has the same asymptotic expression as that of $\mathscr{R}_{i,\mathrm{lb}}^{(1)} (\vec{h}_i)$ in the high SNR regime. As we will see later, the computaion complexity of the lower bound on the outage capacity inspired by $\mathscr{R}^{(2)}_{i,\mathrm{lb}} (\vec{h}_i)$  is much lower than that of $\mathscr{R}^{(1)}_{i,\mathrm{lb}} (\vec{h}_i)$.

Let us consider a ``fair'' FH system in which $v_{i}=v$ for some $1 \leq v \leq u$ and for any $ 1\leq i\leq N$. As explained before, $\mathscr{G}_{i}$ depends on the ordering of $\{c_{i,l}\}_{l=0}^{L_{i}}$, which requires analyzing the order statistics of the channel gains. To avoid this, the following Lemma introduces a lower bound on $\mathscr{G}_{i}$ that only depends on $c_{i,L_{i}}$ i.e., the largest interference crossover gain.
\begin{lem}
\label{dool2}
In a fair FH system,
\begin{equation}
\mathscr{H}_{i}=-(N-1)\left(\frac{v}{u}\log\frac{v}{u}+\left(1-\frac{v}{u}\right)\log\left(1-\frac{v}{u}\right)\right)
\end{equation}
and
\begin{equation}
\mathscr{G}_{i}\geq\mathscr{G}_{i,\mathrm{lb}}\triangleq\frac{\mathrm{E}_B\left\{\log\left(1+\left(1-(\frac{v}{u})^{B}\right)c_{i,L_{i}}\gamma\right)\right\}-(N-1)\frac{v}{u}\log \frac{v}{u}}{c_{i,L_{i}}\gamma+1},\end{equation}
where $B$ is a Binomial random variable with parameters $\left(N-1,\frac{v}{u}\right)$.
\end{lem}
\begin{proof}
Let us define $p\triangleq\frac{v}{u}$. As $L_{i}=2^{N-1}-1$ with probability one and each user selects a certain frequency sub-band with probability $p$. The collection $\{a_{i,l}\}_{l=0}^{L_{i}}$ consists of the numbers $p^{j}(1-p)^{N-1-j}$ repeated ${N-1\choose j}$ times for $0\leq j\leq N-1$.
Hence,
\begin{eqnarray}
\mathscr{H}_{i}&=&-\sum_{j=0}^{N-1}{N-1\choose j}p^{j}(1-p)^{N-1-j}\log\left(p^{j}(1-p)^{N-1-j}\right) \notag\\
&=&-\bigg(\sum_{j=0}^{N-1}j{N-1\choose j}p^{j}(1-p)^{N-1-j}\bigg)\log p \notag\\
&&-\bigg(\sum_{j=0}^{N-1}(N-1-j){N-1\choose j}p^{j}(1-p)^{N-1-j}\bigg)\log(1-p) \notag\\
&=& -(N-1)p\log p-(N-1-(N-1)p)\log(1-p)\notag\\
&=& -(N-1)\left( p \log p +(1-p)\log (1-p) \right).
\end{eqnarray}
As for $\mathscr{G}_{i}$, computation of $\sum_{m=0}^{l-1}a_{i,m}$ is not an easy task. In fact, it depends on the ordering of the crossover gains. For example, if $N=4$,
\begin{equation}
\sum_{m=0}^{3}a_{1,m}=\left\{\begin{array}{ cc}
    (1-p)^3 +2p(1-p)^{2}+p^{2}(1-p)  & \textrm{if $\vert h_{2,1}\vert^{2}<\vert h_{3,1}\vert^{2}<\vert h_{2,1}\vert^{2}+\vert h_{3,1}\vert^{2}<\vert h_{4,1}\vert^{2}$}  \\
   (1-p)^3 + 3p(1-p)^{2}   &  \textrm{if $\vert h_{2,1}\vert^{2}<\vert h_{3,1}\vert^{2}<\vert h_{4,1}\vert^{2}<\vert h_{2,1}\vert^{2}+\vert h_{3,1}\vert^{2}$}\end{array}\right..
\end{equation}
To avoid this difficulty in describing $\mathscr{G}_{i}$, we derive a lower bound on this quantity, which is not sensitive to the ordering of crossover gains. Taking each $a_{i,l}$, there exists a $0\leq s\leq N-1$ such that $a_{i,l}=p^{s}(1-p)^{N-1-s}$. This implies that $a_{i,l}$ corresponds to the interference plus noise power level $\frac{\sum_{j=1}^{s}\vert h_{k_{j},i}\vert^{2}}{v}P+\sigma^{2}$ for some $1\leq k_{1}<\cdots<k_{s}\leq N$ where $k_{j}\neq i$ for $1\leq j\leq s$. Since  $\frac{\sum_{j=1}^{s}\vert h_{k_{j},i}\vert^{2}}{v}P+\sigma^{2}>\frac{\sum_{t\in \mathcal{A}\subsetneq\{1,2,\cdots,s\}}\vert h_{k_{t},i}\vert^{2}}{v}P+\sigma^{2}$ for any set $\mathcal{A}\subsetneq\{1,2,\cdots,s\}$, and $\frac{\sum_{t\in \mathcal{A}\subsetneqq\{1,2,\cdots,s\}}\vert h_{k_{t},i}\vert^{2}}{v}P+\sigma^{2}$ is itself a power level in the PDF of the noise plus interference on each frequency sub-band, we conclude that its associated probability $p^{\vert \mathcal{A}\vert}(1-p)^{N-1-\vert \mathcal{A}\vert}$ is an element in the sequence $(a_{i,0},a_{i,1},\cdots,a_{i,l-1})$. Therefore, we come up with the following lower bound,
\begin{equation}
\label{kol}
\sum_{m=0}^{l-1}a_{i,m}\geq \sum_{\mathcal{A}\subsetneq\{1,2,\cdots,s\}}p^{\vert \mathcal{A}\vert}(1-p)^{N-1-\vert \mathcal{A}\vert}=\sum_{s'=0}^{s-1}{s\choose s'}p^{s'}(1-p)^{N-1-s'}.
\end{equation}
Using (\ref{kol}) in (\ref{ol}) yields
\begin{eqnarray}
\mathscr{G}_{i} &\geq& \frac{1}{c_{i,L_{i}}\gamma+1}\sum_{s=1}^{N-1}{N-1\choose s}p^{s}(1-p)^{N-1-s}\log \left(1+\frac{(c_{i,L_{i}}\gamma+1)\sum_{s'=0}^{s-1}{s\choose s'}p^{s'}(1-p)^{N-1-s'}}{p^{s}(1-p)^{N-1-s}}\right)\notag\\
&=&\frac{1}{c_{i,L_{i}}\gamma+1}\sum_{s=1}^{N-1}{N-1\choose s}p^{s}(1-p)^{N-1-s}\log \left(1+\frac{(c_{i,L_{i}}\gamma+1)\sum_{s'=0}^{s-1}{s\choose s'}p^{s'}(1-p)^{s-s'}}{p^{s}}\right) \notag\\
&=&\frac{1}{c_{i,L_{i}}\gamma+1}\sum_{s=1}^{N-1}{N-1\choose s}p^{s}(1-p)^{N-1-s}\log \left(1+\frac{(1-p^{s})(c_{i,L_{i}}\gamma+1)}{p^{s}}\right) \notag\\
&=& -\frac{1}{c_{i,L_{i}}\gamma+1}\sum_{s=0}^{N-1}{N-1\choose s}sp^{s}(1-p)^{N-1-s}\log p \notag\\
&& +\frac{1}{c_{i,L_{i}}\gamma+1}\sum_{s=0}^{N-1}{N-1\choose s}p^{s}(1-p)^{N-1-s}\log \left(1+(1-p^{s})c_{i,L_{i}}\gamma \right) \notag\\
&=& \frac{\mathrm{E}_B \left\{\log\big(1+(1-p^{B})c_{i,L_{i}}\gamma\big)\right\}-(N-1)p\log p}{c_{i,L_{i}}\gamma+1},
\end{eqnarray}
where $B$ is a Binomial random variable of parameters $(N-1,p)$.
\end{proof}

From now on, we replace $\mathscr{G}_{i}$ with $\mathscr{G}_{i,\mathrm{lb}}$ in all expressions offered for the lower bounds on $\mathscr{R}_{i} (\vec{h}_i)$. In a fair FH system, we denote $a_{i,0}$, $\mathscr{H}_{i}$ and $\mathscr{G}_{i,\mathrm{lb}}$ by $a(v,N)$, $\mathscr{H}(v,N)$ and $\mathscr{G}_{i,\mathrm{lb}}(v,N)$ respectively\footnote{We note that $a_{i,0}$ and $\mathscr{H}_{i}$ do not depend on $i$, however, $\mathscr{G}_{i,\mathrm{lb}}$ depends on $\sum_{j\neq i}|h_{j,i}|^{2}$.}, to emphasize their dependence on the parameters $v,N$.

As a special case of the fair system, let us assume $v_{i}=u$ for all $i$, i.e., all users spread their power on the whole spectrum. This scheme is called full-band spreading (FBS). In this case, it can be observed that $a_{i,l}=0$ for $l\leq L_{i}-1$ and $a_{i,L_{i}}=1$. This yields $a(u,N)=\mathscr{H}(u,N)=\mathscr{G}_{i,\mathrm{lb}}(u,N)=0$. In fact, $\mathscr{R}^{(2)}_{i,\mathrm{lb}}(\vec{h}_{i})$ is tight for $v=u$, i.e., $\mathscr{R}^{(2)}_{i,\mathrm{lb}} (\vec{h}_i)$ is exactly the achievable rate of the $i^{th}$ user while all users transmit over the whole spectrum. We denote this rate by $\mathscr{R}_{i,\mathrm{FBS}} (\vec{h}_i)$, which is given by
\begin{equation}
\label{tg}
\mathscr{R}_{i,\mathrm{FBS}} (\vec{h}_i)=u\log\left(\frac{\vert h_{i,i}\vert^{2}\gamma}{u\left(\frac{c_{i,L_{i}}\gamma}{u}+1\right)}+1\right).
\end{equation}

\textit{Observation 3}- A straightforward method to develop a lower bound on the achievable rate of an additive non-Gaussian noise channel is to replace the noise with a Gaussian noise of the same covariance matrix. Following this approach, it is easy to derive the following lower bound on $\mathscr{R}_{i} (\vec{h}_{i})$,
\begin{equation} \label{Gaussian}
\mathscr{R}_{i,\mathrm{g}}(\vec{h}_{i})\triangleq v\log\left(\frac{\vert h_{i,i}\vert^{2}\gamma}{v\left(\frac{c_{i,L_{i}}\gamma}{u}+1\right)}+1\right)\end{equation}
where the index ``$\mathrm{g}$'' stands for Gaussian. There are two facts that are worth mentioning about $\mathscr{R}_{i,\mathrm{g}}(\vec{h}_{i})$. First, it is seen that $\lim_{\gamma\to\infty}\mathscr{R}_{i,\mathrm{g}}(\vec{h}_{i})<\infty$. Another point is that $\mathscr{R}_{i,\mathrm{g}}(\vec{h}_{i})$ is an increasing function of $v$. However, setting $v=u$ in the expression of $\mathscr{R}_{i,\mathrm{g}}(\vec{h}_{i})$ yields the expression of $\mathscr{R}_{i,\mathrm{FBS}} (\vec{h}_i)$. Therefore, for all realizations of the channel gains and all ranges of $\gamma$,
\begin{equation}
\mathscr{R}_{i,\mathrm{g}}(\vec{h}_{i})\leq\mathscr{R}_{i,\mathrm{FBS}} (\vec{h}_i).\end{equation}
This indicates that using $\mathscr{R}_{i,\mathrm{g}}(\vec{h}_{i})$ as a lower bound on the achievable rate of users in the FH scheme provides no proof of advantage for FH over FBS.

\section{System Design}
  In this section, we aim to find the optimum operation point in the FH scenario in terms of $\epsilon$-outage capacity per user. This requires finding the optimum values of $\{v_i\}_{i=1}^N$. For simplicity of analysis and fairness, we consider the fair system in which $v_i=v$ for any $ 1 \leq i \leq N$. Therefore, the problem is reduced to finding the optimum value of $v$. As mentioned earlier in the system model, we assume the transmitters are not aware of the number of active users as well as the channel gains. Generally, the number of active users in the system is a random variable $N$ with the probability mass function $q_{n}=\Pr\{N=n\}$ for $n\geq 0$. We usually assume $q_{0}=0$ unless otherwise stated.

 Assuming the transmission rate of the $i^{th}$ user is equal to $R$, the outage event for this user is
 \begin{equation}
 \mathcal{O}_{i,\mathrm{FH}} (R)=\left \{ \vec{h}_i: \mathscr{R}_{i,\mathrm{FH}} (\vec{h}_i)<R \right\},
 \end{equation}
where the subscript ``FH'' denotes the underlying scenario\footnote{In the next section, the performance of the FH scheme is compared to that of the FD scheme. We distinguish the parameters of different scenarios  by using the appropriate subscripts.} for which the outage is computed, e.g., Frequency Hopping in this case.
 We notice that the randomness of the number of active users is also considered in the outage event. Therefore, the $\epsilon$-outage capacity\footnote{We are interested in the values of $\epsilon$ in the range $\epsilon<\min\{q_{n}: n\in\mathbb{N}\}$.} of the FH scenario can be expressed as
 \begin{equation}
 \label{heho}
 R_{\mathrm{FH}}(\epsilon)=\sup \left\{R: \Pr\{\mathcal{O}_{i,\mathrm{FH}} (R)\}\leq \epsilon \right\}.
 \end{equation}

The goal of this section is to find $v_{\mathrm{opt}}$ given by
\begin{eqnarray}
 v_{\mathrm{opt}} \triangleq \arg\,\max_{v}\, R_{\mathrm{FH}}(\epsilon).
\end{eqnarray}
We remark that $v_{\mathrm{opt}}$ depends on $\epsilon$, $\gamma$ and $\{q_{n}\}_{n\geq 0}$.

As mentioned in the previous section, the exact expression for $R_{\mathrm{FH}}(\epsilon)$ cannot be derived. This is due to the fact that a closed expression for $ \mathscr{R}_{i,\mathrm{FH}} (\vec{h}_i)$ is intractable. In this part, we derive lower bounds on $R_{\mathrm{FH}}(\epsilon)$ using different lower bounds on  $ \mathscr{R}_{i,\mathrm{FH}} (\vec{h}_i)$ derived in the previous section.

Let $\mathscr{R}_{i,\mathrm{lb}} (\vec{h}_i)$ be a typical lower bound on $\mathscr{R}_{i,\mathrm{FH}} (\vec{h}_i)$ for all realizations of $\vec{h}_i$. It is obvious that $\mathcal{O}_{i,\mathrm{FH}} (R)\subset \{\vec{h}_i: \mathscr{R}_{i,\mathrm{lb}} (\vec{h}_i)<R\}$. This yields $\Pr\{\mathcal{O}_{i,\mathrm{FH}} (R)\}\leq \Pr\big\{\vec{h}_{i}:\mathscr{R}_{i,\mathrm{lb}} (\vec{h}_i)<R\big\}$, and hence,
 \begin{equation}
 \left\{R: \Pr\big\{ \vec{h}_{i}: \mathscr{R}_{i,\mathrm{lb}} (\vec{h}_i)<R\big\}\leq \epsilon\right\}\subset\Big\{R:\Pr\{\mathcal{O}_{i,\mathrm{FH}} (R)\}\leq \epsilon\Big\}.
 \end{equation}
 \label{hen}
 Defining \begin{equation} R_{\mathrm{FH,lb}}(\epsilon)\triangleq\sup\left\{R: \Pr\big\{\vec{h}_{i}: \mathscr{R}_{i,\mathrm{lb}} (\vec{h}_i)<R\}<\epsilon \right\},\end{equation} we get
 \begin{equation}
 R_{\mathrm{FH,lb}}(\epsilon)\leq  R_{\mathrm{FH}}(\epsilon).
  \end{equation}
Based on the preceding discussion, we can derive lower bounds on $R_{\mathrm{FH}}(\epsilon)$, namely $R^{(1)}_{\mathrm{FH,lb}}(\epsilon)$ and $R^{(2)}_{\mathrm{FH,lb}}(\epsilon)$ associated with the lower bounds $\mathscr{R}^{(1)}_{i,\mathrm{lb}} (\vec{h}_i)$ and $\mathscr{R}^{(2)}_{i,\mathrm{lb}} (\vec{h}_i)$ respectively. Consequently, we can obtain estimates of $v_{\mathrm{opt}}$ by maximizing $R^{(1)}_{\mathrm{FH,lb}}(\epsilon)$ or $R^{(2)}_{\mathrm{FH,lb}}(\epsilon)$ over $v$. In the following subsections, we separately compute $R^{(1)}_{\mathrm{FH,lb}}(\epsilon)$ and $R^{(2)}_{\mathrm{FH,lb}}(\epsilon)$.

  \textit{1- Computation of $R^{(1)}_{\mathrm{FH,lb}}(\epsilon)$}

  We start with the following definitions.
  \begin{definition}
  Let $n\in \mathbb{N}$. For $c\in[0,1]$ and $b>0$, $\alpha_{n}(.;b,c):\mathbb{R}^{+}\rightarrow\mathbb{R}^{+}$ is defined by
  \begin{equation}
  \alpha_{n}(\theta;b,c)\triangleq\frac{\mathrm{E}_B\{\log\big(1+b(1-c^{B})\theta\big)\}-(n-1)c\log c}{b\theta+1} , \end{equation}
  where $B$ is a Binomial random variable with parameters $(n-1,c)$.
  \end{definition}
  \begin{definition}
  Let $n\in \mathbb{N}^{\geq 2}$. For $b_{1}<0$, $b_{2}>0$ and $c\in[0,1]$, we define
    \begin{equation}
        \psi_{n}(b_{1},b_{2},c)\triangleq\int_{\theta_{1}\geq0,\cdots,\theta_{n-1}\geq0}\exp\bigg(b_{1}2^{-\alpha_{n}(\theta_{n-1,1};b_{2},c)}\prod_{m=1}^{n-1}\prod_{m'=1}^{{n-1\choose m}}(b_{2}\theta_{m,m'}+1)^{\beta_{m,n}(c)}-\theta_{n-1,1}\bigg)d\theta_{1}\cdots d\theta_{n-1} ,      \end{equation}
        where for each $m$, $\{\theta_{m,m'}\}_{m'=1}^{{n-1\choose m}}$ consists of all possible summations of $m$ elements in the set of dummies $\{\theta_{i}\}_{i=1}^{n-1}$ and $\beta_{m,n}(c)\triangleq c^{m}(1-c)^{n-1-m}$.
  \end{definition}
  For example,
  \begin{equation}
  \psi_{2}(b_{1},b_{2},c)=\int_{0}^{\infty}\exp\bigg(b_{1}2^{-\alpha_{2}(\theta;b_{2},c)}(b_{2}\theta+1)^{c}-\theta\bigg)d\theta,       \end{equation}
  and
  \begin{equation}
   \psi_{3}(b_{1},b_{2},c)=\int_{\theta_{1},\theta_{2}>0}\exp\bigg(b_{1}2^{-\alpha_{3}(\theta_{2,1};b_{2},c)}\big((b_{2}\theta_{1}+1)(b_{2}\theta_{2}+1)\big)^{c(1-c)}(b_{2}\theta_{2,1}+1)^{c^{2}}-\theta_{2,1}\bigg)d\theta_{1}d\theta_{2} ,      \end{equation}
   where $\theta_{2,1}=\theta_{1}+\theta_{2}$ by definition.

  The following Proposition offers an expression to compute $R^{(1)}_{\mathrm{FH,lb}}(\epsilon)$.
   \begin{proposition}
   \begin{equation}
   \label{km}
   R_{\mathrm{FH,lb}}^{(1)}(\epsilon)=\sup\left\{ R:q_{1}\exp\left(\frac{\left(1-2^{\frac{R}{v}}\right)v}{\gamma}\right)+\sum_{n=2}^{\infty}q_{n}\psi_{n}\big(b_{1,n},b_{2},c\big)>1-\epsilon\right\}   \end{equation}
    where $b_{1,n}=\frac{2^{\mathscr{H}(v,n)}\left(1-2^{\frac{R}{v}}\right)v}{\gamma}$, $b_{2}=\frac{\gamma}{v}$ and $c=\frac{v}{u}$.
    \end{proposition}
  \begin{proof}
  See appendix B.
  \end{proof}
  The expression given in (\ref{km}) is quite complicated. On one hand, the multiple integrals do not have a closed form. On the other hand, the maximization $\max_{v}R_{\mathrm{FH,lb}}^{(1)}(\epsilon)$ must be computed numerically. However, $R^{(1)}_{\mathrm{FH,lb}}(\epsilon)$ is the best lower bound on $R_{\mathrm{FH}}(\epsilon)$ as $\mathscr{R}^{(1)}_{i,\mathrm{lb}}(\vec{h}_{i})$ is the best lower bound we have found on the achievable rate of the $i^{th}$ user in the FH scenario.

 \textit{2- Computation of $R^{(2)}_{\mathrm{FH,lb}}(\epsilon)$}

 We start with the following definition.
 \begin{definition}
 Let $n\in\mathbb{N}\backslash\{1\}$. For $b_{1}<0$, $b_{2}>0$ and $c_{1},c_{2}\in [0,1]$, we define the function $\phi_{n}(b_{1},b_{2},c_{1},c_{2})$ as
  \begin{equation}
  \label{lk}
  \phi_{n}(b_{1},b_{2},c_{1},c_{2})\triangleq\frac{1}{(n-2)!}\int_{0}^{\infty}\theta^{n-2}\exp\big(b_{1}(b_{2}\theta+1)^{c_{1}}2^{-\alpha_{n}(\theta;b_{2},c_{2})}-\theta\big)d\theta.
  \end{equation}
  \end{definition}
  Using this class of functions, the following proposition yields $R^{(2)}_{\mathrm{FH,lb}}(\epsilon)$.
  \begin{proposition}
  \begin{equation}
  \label{pqqq}
  R^{(2)}_{\mathrm{FH,lb}}(\epsilon)=\sup \left\{R: q_{1}\exp\left(\frac{\left(1-2^{\frac{R}{v}}\right)v}{\gamma}\right)+\sum_{n=2}^{\infty}q_{n}\phi_{n}\big(b_{1,n},b_{2},c_{1,n},c_{2}\big)>1-\epsilon\right\},  \end{equation}
  where $b_{1,n}=\frac{2^{\mathscr{H}(v,n)}\left(1-2^{\frac{R}{v}}\right)v}{\gamma}$, $b_{2}=\frac{\gamma}{v}$, $c_{1,n}=1-a(v,n)$ and $c_{2}=\frac{v}{u}$.
  \end{proposition}
  \begin{proof}
  See appendix C.
  \end{proof}
  Comparing the expressions for $R^{(1)}_{\mathrm{FH,lb}}(\epsilon)$ and $R^{(2)}_{\mathrm{FH,lb}}(\epsilon)$, it can be observed that computation of $R^{(2)}_{\mathrm{FH,lb}}(\epsilon)$ involves only one integration, while the computation of $R^{(1)}_{\mathrm{FH,lb}}(\epsilon)$ involves multiple integrations that are not tractable for many cases. To further reduce the complexity of computation, the following Corollary, proved in appendix D, yields another lower bound on $R_{\mathrm{FH}}(\epsilon)$ that involves no numerical integrations. We denote this lower bound by $R^{(3)}_{\mathrm{FH,lb}}(\epsilon)$.
  \begin{corollary}
  Let
    \begin{equation} \label{kuyyywh0}
 R^{(3)}_{\mathrm{FH,lb}}(\epsilon)\triangleq\sup\left\{R: \sum_{n=1}^{\infty}q_n \exp\left(b_{1,n}\big((n-1)b_{2}+1\big)^{1-a(v,n)}\right)>1-\epsilon\right\}.\end{equation}
 Then,
  \begin{equation}R_{\mathrm{FH}}(\epsilon)\geq R^{(1)}_{\mathrm{FH,lb}}(\epsilon)\geq R^{(2)}_{\mathrm{FH,lb}}(\epsilon)\geq R^{(3)}_{\mathrm{FH,lb}}(\epsilon),\end{equation}
 where $b_{1,n}=\frac{2^{\mathscr{H}(v,n)}\left(1-2^{\frac{R}{v}}\right)v}{\gamma}$ and $b_{2}=\frac{\gamma}{v}$.
  \end{corollary}
  \begin{proof}
  See appendix D.
    \end{proof}
       Fig. 1 shows the three lower bounds $\max_{v}R_{\mathrm{FH,lb}}^{(k)}(\epsilon)$ for $1\leq k\leq3$ on $\max_{v}R_{\mathrm{FH}}(\epsilon)$ in a system where, at most, four users become active simultaneously with $(q_{1},q_{2},q_{3},q_{4})=(0.4,0.2,0.2,0.2)$, $u=8$ and $\gamma=100$. As can be observed from this figure,  $\max_{v}R_{\mathrm{FH,lb}}^{(2)}(\epsilon)$ and $\max_{v}R_{\mathrm{FH,lb}}^{(3)}(\epsilon)$ are pretty close to each other while they are within a considerable gap to $\max_{v}R_{\mathrm{FH,lb}}^{(1)}(\epsilon)$, especially for larger values of $\epsilon$. It is notable that the maximization over $v$ is performed separately for each $\epsilon$.
   \begin{figure}[h!b!t]
  \centering
  \includegraphics[scale=.6] {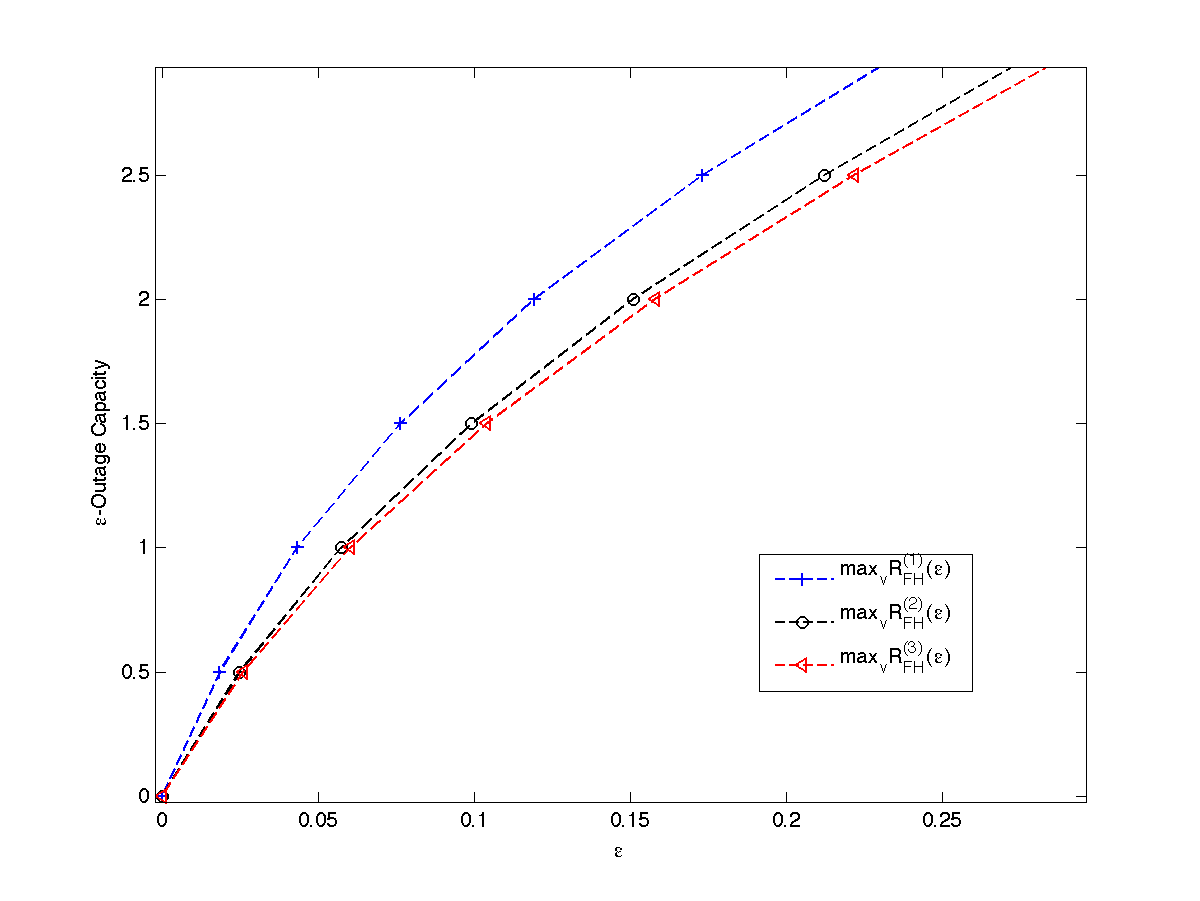}
  \caption{Depictions of $\max_{v}R^{(k)}_{\mathrm{FH,lb}}(\epsilon)$ for $1\leq k\leq3$ in a setup where $(q_{1},q_{2},q_{3},q_{4})=(0.4,0.2,0.2,0.2)$, $u=8$ and $\gamma=100$.}
  \label{figapp1}
 \end{figure}

Having the expression for the $\epsilon$-outage capacity, we can find the best operational point of the system in terms of $v$, the number of selected sub-bands. For this purpose, we consider some asymptotic cases in terms of $\epsilon$ and $\gamma$ and discuss the optimum value of $v$ in these regimes. As the expression of $R_{\mathrm{FH,lb}}^{(1)}(\epsilon)$ is not analytically tractable, we use $R_{\mathrm{FH,lb}}^{(3)}(\epsilon)$ for our analysis\footnote{Note that we can also use $R_{\mathrm{FH,lb}}^{(2)}(\epsilon)$ since it involves only one integration, however, as $R_{\mathrm{FH,lb}}^{(3)}(\epsilon)$ is pretty close to $R_{\mathrm{FH,lb}}^{(2)}(\epsilon)$, we use $R_{\mathrm{FH,lb}}^{(3)}(\epsilon)$ for simplicity of analysis.}. For simulation purposes, we use the most complex lower bound $R_{\mathrm{FH,lb}}^{(1)}(\epsilon)$, which is the best bound as well.
\subsection{Asymptotically small $\epsilon$}

In this case, one can easily show that  $\lim_{\epsilon\to0}R^{(3)}_{\mathrm{FH,lb}}(\epsilon) = 0$. Therefore, in (\ref{kuyyywh0}), we can approximate $\frac{\left(1-2^{\frac{R}{v}}\right)v}{\gamma}$ by $-\frac{R \ln 2}{\gamma}$ and hence,
\begin{eqnarray} \label{kuyyywh1}
\exp\left(\frac{\left(1-2^{\frac{R}{v}}\right)v}{\gamma}\right) &\approx& 1- \frac{R \ln (2)}{\gamma}.
\end{eqnarray}
By the same token, the term on the right-hand side of (\ref{kuyyywh0}) can be approximated as
\begin{eqnarray} \label{kuyyywh2}
 \sum_{n=1}^{\infty}q_n \exp\left(b_{1,n}\big((n-1)b_{2}+1\big)^{1-a(v,n)}\right) &\approx&  \sum_{n=1}^{\infty}q_n \exp\left( - \frac{R \ln 2}{\gamma} 2^{\mathscr{H}(v,n)} \left(\frac{n-1}{v}\gamma+1\right)^{1-a(v,n)}\right) \notag\\
&\stackrel{(a)}{\approx}& \sum_{n=1}^{\infty}q_n \left( 1- \frac{R \ln 2}{\gamma} f(v,n,\gamma) \right),
\end{eqnarray}
where $(a)$ follows from the fact that the term
\begin{eqnarray} \label{kuyyywh3}
f(v,n,\gamma) \triangleq 2^{\mathscr{H}(v,n)} \left(\frac{n-1}{v}\gamma+1\right)^{1-a(v,n)}
\end{eqnarray}
 does not depend on $\epsilon$. Using (\ref{kuyyywh2}) in (\ref{kuyyywh0}) yields
\begin{eqnarray} \label{kuyyywh7}
R^{(3)}_{\mathrm{FH,lb}}(\epsilon) \approx \frac{\epsilon \gamma}{ \Big( \sum_{n=1}^{\infty} q_n f(v,n,\gamma) \Big)\ln2}.
\end{eqnarray}
One can observe that the function $f(v,n,\gamma)$ is concave in terms of $v$ for all $n\geq 1$. Hence, the function $g(v)\triangleq\sum_{n=1}^{\infty} q_nf(v,n,\gamma)$ is concave as well. As such, the minimum of $g(v)$ occurs either at $v=1$ or $v=u$. In partiular, for $v=1$ to be the optimum value, we must have $g(1)<g(u)$ or equivalently,
\begin{equation}
\label{hnmj}
 \sum_{n=1}^{\infty}q_n \left(\frac{u}{(u-1)^{1-\frac{1}{u}}}\right)^{n-1}\big((n-1)\gamma+1\big)^{1-\left(1-\frac{1}{u}\right)^{n-1}}<\sum_{n=1}^{\infty}q_n\left(1+\frac{(n-1)\gamma}{u}\right)=1+\frac{\mathrm{E}\{N\}-1}{u}\gamma ,
\end{equation}
where we have used the fact that $2^{\mathscr{H}\left(1,n\right)}=\left(\frac{u}{(u-1)^{1-\frac{1}{u}}}\right)^{n-1}$.

\textit{Example 1}- Assume $q_{1}=q_{2}=0.5$. The condition (\ref{hnmj}) can be written as
\begin{equation}
\label{hjkl}
\frac{u^{u}}{(u-1)^{u-1}}<\frac{1}{1+\gamma}\left(1+\frac{\gamma}{u}\right)^{u}.
\end{equation}
For each $u$, there is a smallest number $\gamma_{u}>0$ such that if $\gamma>\gamma_{u}$, then (\ref{hjkl}) is satisfied. For example, if $u=2$, then $\gamma_{2}=2\left(3+2\sqrt{3}\right)$. It is easily seen through simulations that $\gamma_{10}<\gamma_{9}<\cdots<\gamma_{3}<\gamma_{2}$. However, the sequence $\{\gamma_{u}\}_{u=1}^{\infty}$ is not a decreasing sequence. In fact, one observes that the right-hand side of (\ref{hjkl}) tends to the increasing function $\frac{e^{\gamma}}{1+\gamma}$ in terms of $\gamma$ as $u$ increases. On the other hand, the left-hand side of the same equation, i.e., the term $\frac{u^{u}}{(u-1)^{u-1}}$, tends to $eu$ as $u$ increases. Therefore, in the case that both $\gamma$ and $u$ tend to infinity, the condition (\ref{hjkl}) reduces to $eu < \frac{e^{\gamma}}{1+\gamma}$, which implies that $\gamma_{u} \approx \ln u + \ln \ln u +1$. This is an increasing function in terms of $u$.
$\rightmark{\square}$

Fig. 2 offers the curves of $R^{(1)}_{\mathrm{FH,lb}}(\epsilon)$ in terms of $v$ for $\epsilon=0.01,0.05$ and $0.15$. The underlying network is characterized with $q_{1}=q_{2}=0.5$, $u=10$ and $\gamma=20\mathrm{dB}$. It is seen that for $\epsilon=0.01$, taking $v=1$ yields the best performance. As already stated in example 1, as far as $q_{1}=q_{2}=0.5$, for sufficiently small $\epsilon$, the equation $\frac{1}{1+\gamma_{10}}\left(1+\frac{\gamma_{10}}{10}\right)^{10}=\frac{10^{10}}{9^{9}}=25.8117$ yields $\gamma_{10}=7.052$ as the minimum SNR value such that choosing $\gamma>\gamma_{10}$ results in $v=1$ as the best choice. Since $\gamma_{10} < 20 \mathrm{dB}$, we expect $v_{\mathrm{opt}}=1$, which is confirmed in the plot. However, as the value of $\epsilon$ increases, we are moving away from the \textit{asymptotically small $\epsilon$} region and $v=1$ is no longer an optimal choice.
 \begin{figure}[h!b!t]
  \centering
  \includegraphics[scale=.6] {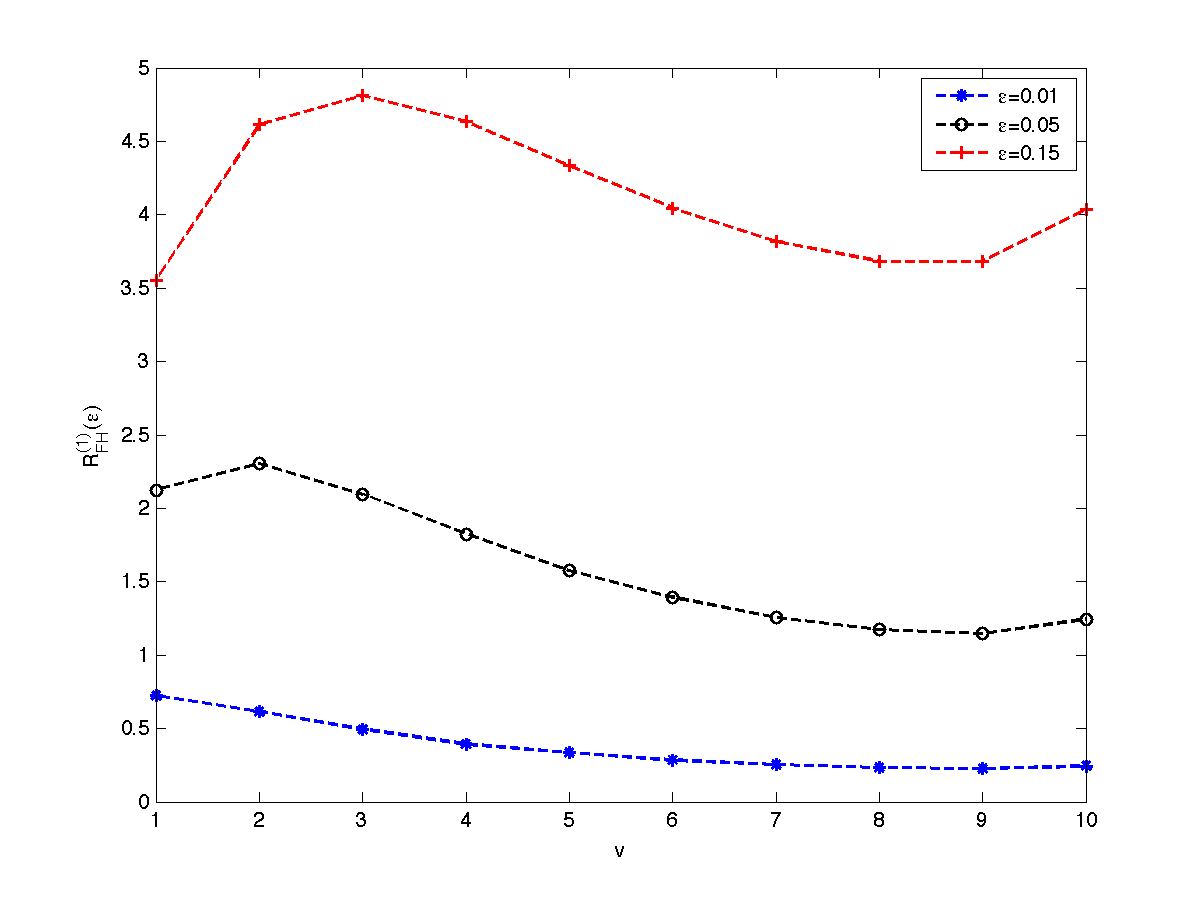}
  \caption{Sketch of $R_{\mathrm{FH,lb}}^{(1)}(\epsilon)$  for different values of $\epsilon$ in a network with $(q_{1},q_{2})=(0.5,0.5)$, $u=10$ and $\gamma=20\mathrm{dB}$.  }
  \label{figapp1}
 \end{figure}
\subsection{Asymptotically small $\gamma$}
In this case, one can easily show that $\lim_{\gamma \to 0} R^{(3)}_{\mathrm{FH,lb}}(\epsilon)=0$. Therefore, similar to the previous case, we can use the approximation
\begin{eqnarray}
\frac{\left(1-2^{\frac{R}{v}}\right) v}{\gamma} \approx -\frac{R \ln 2}{\gamma}.
\end{eqnarray}
Defining $\tau \triangleq \exp \left( -\frac{R \ln 2}{\gamma} \right) $, one can rewrite  (\ref{kuyyywh0}) as
\begin{eqnarray}
R^{(3)}_{\mathrm{FH,lb}}(\epsilon) \approx \sup\left\{R: \sum_{n=1}^{\infty} q_n \tau^{f(v,n,\gamma)} > 1-\epsilon \right\},
\end{eqnarray}
 where $f(v,n,\gamma)$ is given in (\ref{kuyyywh3}).
As $\tau$ is a decreasing function in terms of $R$, we can write
\begin{eqnarray}
R^{(3)}_{\mathrm{FH,lb}}(\epsilon) &=& -\frac{\gamma \ln (\tau^*)}{\ln 2},
\end{eqnarray}
where
\begin{eqnarray} \label{kuyyywh5}
 \tau^* = \inf \left\{ \tau:  \sum_{n=1}^{\infty} q_n \tau^{f(v,n,\gamma)} > 1-\epsilon \right\}.
\end{eqnarray}
Therefore,
\begin{eqnarray}
 \max_v R^{(3)}_{\mathrm{FH,lb}}(\epsilon)  = -\frac{\gamma}{\ln 2} \ln \left( \min_v \tau^* \right).
\end{eqnarray}
It can be shown that $\min_v \tau^*$ occurs when $f(v,n,\gamma)$ takes its minimum value over $v$ for each $n$. But, as $\gamma \ll 1$, the term $\big(\frac{n-1}{v}\gamma+1\big)^{1-a(v,n)} \approx 1$ and hence, $f(v,n,\gamma) \approx 2^{\mathscr{H} (v,n)}$, which is uniformly minimized for all values of $n$ by taking $v=u$. This gives $\tau^* = 1-\epsilon$, and hence,
\begin{eqnarray} \label{kuyyywh4}
 R^{(3)}_{\mathrm{FH,lb}}(\epsilon) = - \gamma \log (1-\epsilon).
\end{eqnarray}
This is exactly the outage capacity of a point-to-point system without interference. Therefore, in the low SNR regime, interference has no destructive effect on the outage capacity justifying the optimality of $v=u$.

Fig. \ref{kuyyyw1} presents the plot of $R_{\mathrm{FH,lb}}^{(1)}(0.1)$ versus $v$ for $\gamma = -10 \mathrm{dB}$ in a system with $q_{1}=q_{2}=0.5$ and $u=10$.  It is seen that $v_{\mathrm{opt}}=u=10$, which is expected by our analysis. Also, we obseve that the outage capacity is not quite sensitive to the value of $v$.

\begin{figure}[h!b!t]
  \centering
  \includegraphics[scale=.6] {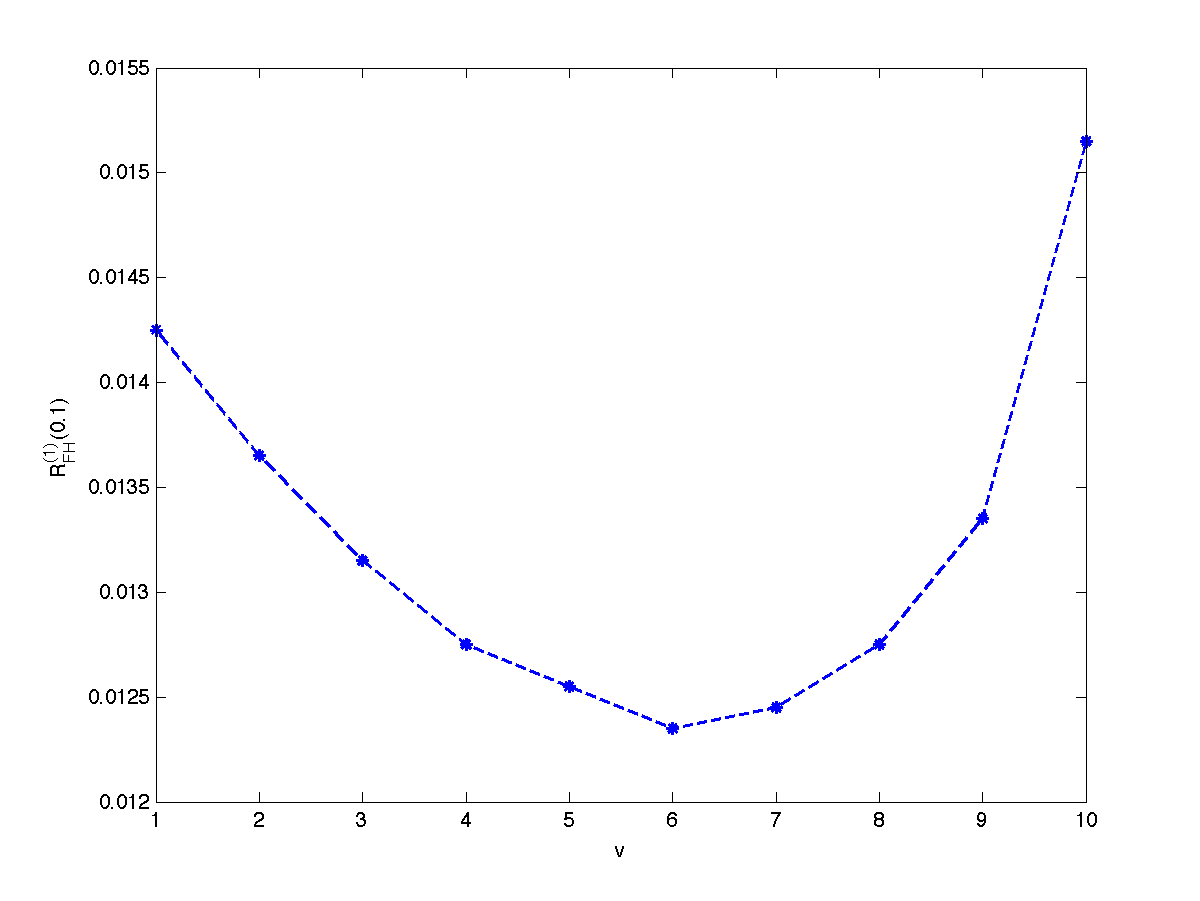}
  \caption{Sketch of $R_{\mathrm{FH,lb}}^{(1)}(0.1)$ in a network with $(q_{1},q_{2})=(0.5,0.5)$, $\gamma=-10\mathrm{dB}$ and $u=10$.  }
  \label{kuyyyw1}
 \end{figure}
\subsection{Asymptotically high $\gamma$}

Recalling the expression of $ R^{(3)}_{\mathrm{FH,lb}}(\epsilon)$ given in (\ref{kuyyywh0}), we have
\begin{eqnarray} \label{kuyyywh6}
R^{(3)}_{\mathrm{FH,lb}}(\epsilon) = \sup  \left\{ R: \sum_{n=1}^{\infty} q_n \exp \left(- \frac{vf(v,n,\gamma)\left(2^{R/v}-1\right)}{\gamma}\right) > 1-\epsilon \right\}.
\end{eqnarray}
 As $\gamma \to \infty$, the term $f(v,n,\gamma)$ grows polynomially with $\gamma$ for any $n\geq 1$. Assuming there exists $n_{\max}$ such that $\Pr \{N > n_{\max}\}=0$, the outage event is determined by the term with the maximum $f(v,n,\gamma)$. In fact, since $f(v,n,\gamma)$ is an increasing function in terms of $n$, it follows that $\sum_{n=1}^{\infty} q_n \exp \left(- \frac{vf(v,n,\gamma)\left(2^{R/v}-1\right)}{\gamma}\right)\approx \sum_{n=1}^{n_{\max}-1}q_{n}+q_{n_{\max}}\exp \left(- \frac{v f(v,n_{\max},\gamma) \left(2^{R/v}-1\right)}{\gamma} \right)$. Therefore,  (\ref{kuyyywh6}) simplifies to
\begin{eqnarray}
R^{(3)}_{\mathrm{FH,lb}}(\epsilon) \approx\sup  \left\{ R: \exp \left(- \frac{v f(v,n_{\max},\gamma) \left(2^{R/v}-1\right)}{\gamma} \right) > 1-\frac{\epsilon}{q_{n_{\max}}} \right\},
\end{eqnarray}
which gives
\begin{eqnarray}
R^{(3)}_{\mathrm{FH,lb}}(\epsilon) &\approx& v \log \left( 1- \frac{\gamma \ln \left( 1-\frac{\epsilon}{q_{n_{\max}}}\right)}{v f(v,n_{\max},\gamma)}\right) \notag\\
&=& v \log \left( 1- \frac{\gamma 2^{-\mathscr{H}(v,n_{\max})}  \ln \left( 1-\frac{\epsilon}{q_{n_{\max}}}\right)}{v \left(1+\frac{n_{\max}-1}{v}\gamma \right)^{1-\left(1-\frac{v}{u}\right)^{n_{\max}-1}}}\right) \notag\\
&\stackrel{(a)}{\approx}& v \log \left( 1- \kappa_v  \gamma^{\left(1-\frac{v}{u}\right)^{n_{\max}-1}} \ln \left( 1-\frac{\epsilon}{q_{n_{\max}}}\right) \right), \notag\\
  \label{kuyyywh10}
\end{eqnarray}
where
\begin{equation}\kappa_v \triangleq \frac{2^{-\mathscr{H}(v,n_{\max})}}{v}\left(\frac{n_{\max}-1}{v}\right)^{\left(1-\frac{v}{u}\right)^{n_{\max}-1}-1}\end{equation} and $(a)$ follows from the assumption that $\gamma$ lies in the high SNR range. It is observed that the maximization of the above expression with respect to $v$ is equivalent to the maximization of the term $v \left(1-\frac{v}{u}\right)^{n_{\max}-1}$ with respect to $v$, as SNR tends to infinity. This yields
\begin{eqnarray} \label{vopt}
 v_{\mathrm{opt}} = \left\lceil \frac{u}{n_{\max}} \right\rceil.
\end{eqnarray}

Fig. \ref{kuyyyw2} shows the curves of $R_{\mathrm{FH,lb}}^{(1)}(0.1)$ versus $v$ for different values of SNR in a system with parameters  $n_{\max}=2$, $(q_{1},q_{2})=(0.5,0.5)$ and $u=10$.  In general, for a sufficiently large, however finite, value of SNR, one can obtain the optimum value for $v$ by
\begin{eqnarray}
v_{\mathrm{opt}} = \arg \, \, \max_{v} \, \left\{ v \log \left( 1- \kappa_v  \gamma^{\left(1-\frac{v}{u}\right)^{n_{\max}-1}} \ln \left( 1-\frac{\epsilon}{q_{n_{\max}}}\right) \right)\right\}.
\end{eqnarray}
For example, it is easy to verify that in a system with the above parameters at $\gamma = 40 \mathrm{dB}$, one gets $v_{\mathrm{opt}}=4$, while $\left\lceil \frac{u}{n_{\max}} \right\rceil=\left\lceil \frac{10}{2} \right\rceil=5$. This is in agreement with the plot of $R_{\mathrm{FH,lb}}^{(1)}(0.1)$ given in fig. \ref{kuyyyw2} for $\gamma = 40 \mathrm{dB}$.

\begin{figure}[h!b!t]
  \centering
  \includegraphics[scale=.6] {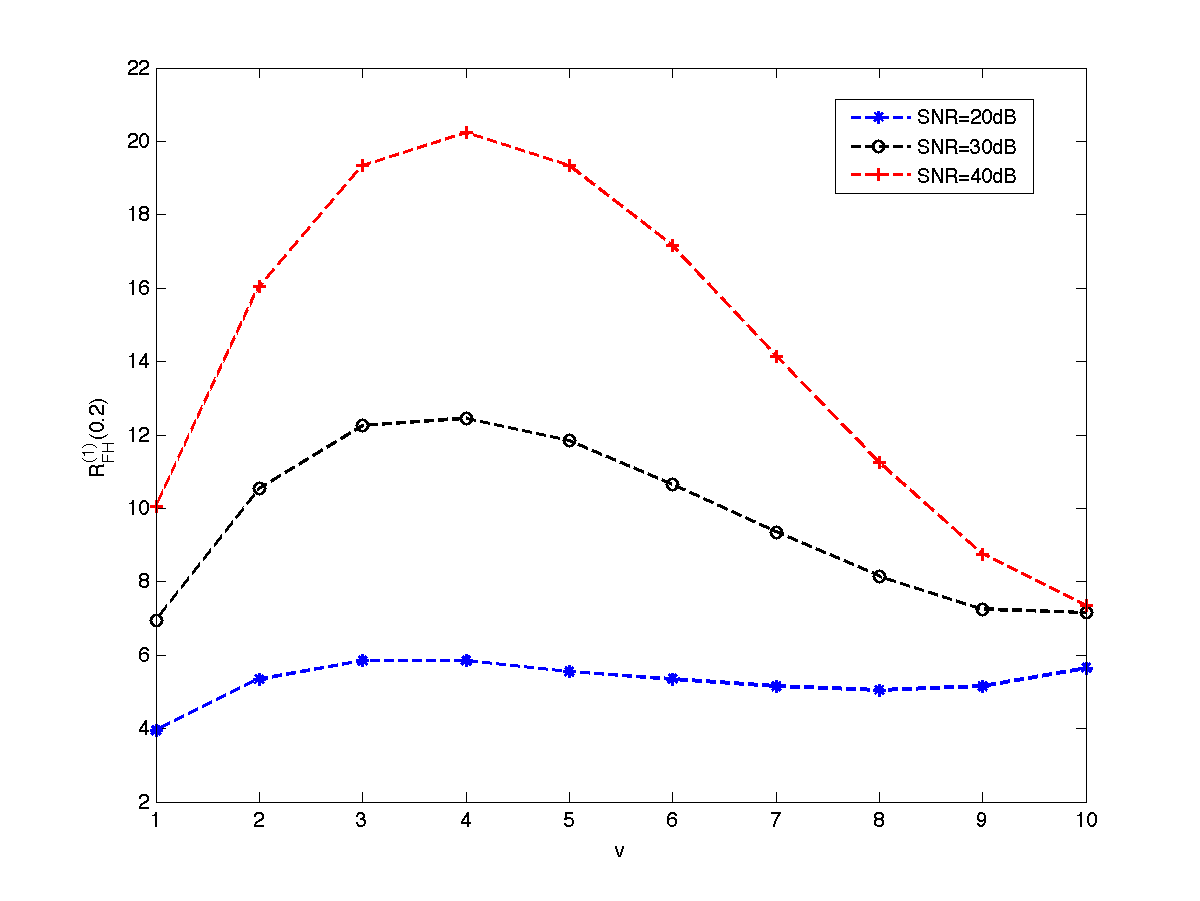}
  \caption{Sketch of $R_{\mathrm{FH,lb}}^{(1)}(0.1)$ at different SNR levels in a network with $n_{\max}=2$, $(q_{1},q_{2})=(0.5,0.5)$ and $u=10$.  }
  \label{kuyyyw2}
 \end{figure}

 \textit{Remark -} In \cite{kamyar1}, we introduced a generalized version of the FH scheme (GFH scheme) where the number of selected sub-bands by each user changes independently from transmission to transmission. It is clear that the achievable rate of this generalized scheme can be higher than the case where all users only hop over a fixed number of  frequency sub-bands. However, unlike \cite{kamyar1} where we could analytically derive the optimum hopping pattern in the generalized FH scenario for the average sum-rate in the high SNR regime, it is not possible to find a clean mathematical formulation (or tight bounds) for outage capacity in the same scenario.

  \section{Comparison with other Schemes}
In this section, we compare the performance of the proposed FH scenario with that of the FD scheme in terms of the $\epsilon$-outage capacity. Frequency Division is a well-known and simple resource allocation scheme that is widely used in the current wireless systems. Based on the number of existing licensees requesting service in the FD scenario, denoted by  $n_{\mathrm{des}}$, the spectrum is primarily divided  into $n_{\mathrm{des}}$ bands\footnote{Each band might consist of several sub-bands. Also, it is assumed that $n_{\mathrm{des}}$ divides $u$.} and each licensed user only occupies one band upon activation.
In the case that all the licensed users are active all the time, i.e., $N=n_{\mathrm{des}}$, this scheme results in the most efficient usage of the bandwidth. This makes the FD scenario superior to other schemes proposed in the literature, especially in the high SNR regime. However, in a practical situation, the number of concurrently active users is much smaller than $n_{\mathrm{des}}$. This makes FD  highly inefficient on the heels that a considerable portion of the sub-bands is unused. In the simulations at the end of this section, we make an assumption that the number of active users $N$ does not exceed $n_{\max}$, almost surely. However, $n_{\max}$ is strictly less\footnote{Usually, $n_{\max} \ll n_{\mathrm{des}}$.}  than $n_{\mathrm{des}}$. In addition to FD, we also study the $\epsilon$-outage capacity of the FBS scenario, which is a spcial case of FH. In fact, FD and FBS can be considered as two extreme spectrum management schemes where the former avoids any interference among the users, while the latter makes all users share the same spectrum all the time. In the sequel, we compute  $R_{\mathrm{FD}}(\epsilon)$ and $R_{\mathrm{FBS}}(\epsilon)$.

\subsection{ Computation of $R_{\mathrm{FD}}(\epsilon)$}

In the FD scenario, the spectrum is already divided into $n_{\mathrm{des}}$ non-overlaping bands each containing $\frac{u}{n_{\mathrm{des}}}$ sub-bands. Each user that becomes active occupies one of the units. As there is no interference among users, the outage event for the $i^{th}$ user can be written as
\begin{equation}
\mathcal{O}_{i,\mathrm{FD}} (R)=\left\{h_{i,i}: \frac{u}{n_{\mathrm{des}}}\log \left(1+\frac{n_{\mathrm{des}}\vert h_{i,i}\vert^{2}\gamma}{u} \right)<R\right\}.
\end{equation}
As $h_{i,i}$ is a complex Gaussian random variable with variance $\frac{1}{2}$ per dimension, $|h_{i,i}|^{2}$
 is an exponential random variable with parameter one. Thus,
\begin{equation}
\Pr\{\mathcal{O}_{i,\mathrm{FD}} (R)\}=1-\exp\left(\frac{(1-2^{\frac{n_{\mathrm{des}}R}{u}})u}{n_{\mathrm{des}}}\frac{1}{\gamma}\right),
\end{equation}
and
\begin{eqnarray}
\label{pq}
R_{\mathrm{FD}}(\epsilon)&=&\sup\left\{R: \Pr\{\mathcal{O}_{i,\mathrm{FD}} (R)\}<\epsilon\right\}\notag\\&=&\sup\left\{R: \exp\left(\frac{(1-2^{\frac{n_{\mathrm{des}}R}{u}})u}{n_{\mathrm{des}}}\frac{1}{\gamma}\right)>1-\epsilon\right\}\notag\\&=&\frac{u}{n_{\mathrm{des}}}\log\left(1-\frac{n_{\mathrm{des}}\gamma\ln(1-\epsilon)}{u}\right).
\end{eqnarray}

\subsection{ Computation of $R_{\mathrm{FBS}}(\epsilon)$}

  Using (\ref{tg}), the following Proposition yields $ R_{\mathrm{FBS}}(\epsilon)$.
  \begin{proposition}

  \begin{equation}
  \label{pqq}
  R_{\mathrm{FBS}}(\epsilon)=\sup\left\{R: \exp\left(\frac{(1-2^{\frac{R}{u}})u}{\gamma}\right)\sum_{n=1}^{\infty}q_{n}2^{-\frac{(n-1)R}{u}}>1-\epsilon\right\}.
  \end{equation}
  \end{proposition}
  \begin{proof}
  See appendix E.
  \end{proof}

\subsection{Asymptotic Comparison of $R_{\mathrm{FH}}(\epsilon)$, $R_{\mathrm{FD}}(\epsilon)$, and $R_{\mathrm{FBS}}(\epsilon)$}
\subsubsection{Asymptotically small $\epsilon$}

Noting that $\ln(1-\epsilon)\approx-\epsilon$, we get $\log\left(1-\frac{n_{\mathrm{des}}\gamma\ln(1-\epsilon)}{u}\right)\approx \frac{n_{\mathrm{des}}\gamma\epsilon}{u\ln2}$. Therefore,
\begin{equation}
 R_{\mathrm{FD}}(\epsilon)\approx\frac{\gamma\epsilon}{\ln2},
\end{equation}
which is the maximum achievable $\epsilon$-outage capacity in the underlying network for small values of $\epsilon$. As for FBS, using (\ref{pqq}) and noting that $\frac{R}{u} \ll 1$, we get
\begin{eqnarray}
R_{\mathrm{FBS}}(\epsilon) &\approx&\sup\left\{R: \left(1-\frac{R \ln 2}{\gamma}\right)\sum_{n=1}^{\infty}q_{n} \left(1-\frac{(n-1)R \ln 2}{u} \right)>1-\epsilon\right\} \notag\\
&\stackrel{(a)}{\approx}& \sup\left\{R: 1-\frac{R \ln 2}{\gamma} -\left(\sum_{n=1}^{\infty}\frac{q_n(n-1)}{u}\right)R\ln2 >1-\epsilon\right\} \notag\\
&=& \frac{\epsilon \gamma}{\left( 1+ \gamma \sum_{n=1}^{\infty} \frac{q_n(n-1)}{u}\right)\ln2} \notag\\
&=& \frac{\epsilon \gamma}{\left( 1+  \frac{\mathrm{E}\{N\}-1}{u}\gamma\right)\ln2} \label{kuyyywh8}
\end{eqnarray}
where in $(a)$ we have neglected terms proportionate to $R^{2}$.
Comparing (\ref{kuyyywh7}) and (\ref{kuyyywh8}) reveals that $R_{\mathrm{FH,lb}}^{(3)} (\epsilon)$ is  larger than $R_{\mathrm{FBS}}(\epsilon)$ in the low $\epsilon$ regime as far as (\ref{hnmj}) is satisfied. Furthermore, in the case that $\mathrm{E}\{N\} \ll \frac{u}{\gamma}$, the above equation implies that the FBS scheme, and consequently the FH scenario, achieves the optimal performance of the FD scheme.

\subsubsection{Asymptotically small $\gamma$}
Using (\ref{pq}) and (\ref{pqq}), it can be realized that for the asymptotically small $\gamma$,
\begin{eqnarray}
 R_{\mathrm{FD}}(\epsilon) = R_{\mathrm{FBS}}(\epsilon) = -\gamma\log (1-\epsilon),
\end{eqnarray}
which is the same value obtained in the previous section for $R_{\mathrm{FH}}(\epsilon)$ (by setting $v=u$) and is the maximum achievable outage capacity in the network. Therefore, in this regime, spectrum division and spectrum sharing both achieve the optimal performance. Furthermore,  the Gaussian lower bound given in (\ref{Gaussian}) implies that the FH scheme is optimal in the low SNR regime, regardless of $v$. However, if we use the proposed lower bounds (e.g., $R_{\mathrm{FH}}^{(3)} (\epsilon)$), due to the presence of  the term $2^{\mathscr{H}(v,n)}$ in the expression of such lower  bounds, taking $v < u$ does not yield this conclusion.

\subsubsection{Asymptotically high $\gamma$}
From (\ref{pq}),
\begin{eqnarray}
R_{\mathrm{FD}}(\epsilon) &\approx& \frac{u}{n_{\mathrm{des}}} \log \gamma + \frac{u}{n_{\mathrm{des}}}\log \left( -\frac{n_{\mathrm{des}}}{u} \ln (1-\epsilon)\right)\notag\\
&\approx& \frac{u}{n_{\mathrm{des}}} \log \gamma,
\label{kuyyywh9}
\end{eqnarray}
 as $\gamma \to \infty$. Also, (\ref{pqq}) indicates that $R_{\mathrm{FBS}}(\epsilon)$ saturates as $\gamma$ tends to infinity and as such, FBS is highly inefficient in this regime. In fact,
    \begin{equation}
 \lim_{\gamma\to \infty} R_{\mathrm{FBS}}(\epsilon)=\sup\left\{R: \sum_{n=1}^{\infty}q_{n}2^{-\frac{(n-1)R}{u}}>1-\epsilon\right\}.
  \end{equation}
  For example, in a system where $n_{\max}=3$ and $(q_{1},q_{2},q_{3})=(0.5,0.3,0.2)$,
    \begin{eqnarray}
 \lim_{\gamma\to \infty} R_{\mathrm{FBS}}(\epsilon)=\sup\left\{R: 0.5+0.3\times2^{-\frac{R}{u}}+0.2\times2^{-\frac{2R}{u}}>1-\epsilon\right\}=u\log\frac{4}{\sqrt{49-80\epsilon}-3}
  \end{eqnarray}
  for any $\epsilon\leq 0.2$.

 In case of FH, (\ref{kuyyywh10}) yields,
\begin{eqnarray}
 \max_{v}R^{(3)}_{\mathrm{FH,lb}}(\epsilon) &\geq& \left \lceil \frac{u}{n_{\max}}\right \rceil \left( 1- \frac{\left \lceil \frac{u}{n_{\max}}\right \rceil}{u}\right)^{n_{\max}-1}\log \gamma+\left \lceil \frac{u}{n_{\max}}\right \rceil \log \left(-\kappa_{\left \lceil \frac{u}{n_{\max}}\right \rceil } \ln \left( 1-\frac{\epsilon}{q_{n_{\max}}}\right) \right). \notag\\
\end{eqnarray}
For simplicity, let us assume that $n_{\max}$ divides $u$. Then,
\begin{eqnarray} \label{kuyyywh11}
\max_{v}R^{(3)}_{\mathrm{FH,lb}}(\epsilon) \geq \frac{u}{n_{\max}} \left( 1-\frac{1}{n_{\max}} \right)^{n_{\max}-1} \log \gamma+ \frac{u}{n_{\max}} \log \left(-\kappa_{\frac{u}{n_{\max}}} \ln \left( 1-\frac{\epsilon}{q_{n_{\max}}}\right) \right).
\end{eqnarray}
Comparing (\ref{kuyyywh9}) and (\ref{kuyyywh11}), we obtain
\begin{eqnarray}
 \lim_{\gamma \to \infty} \frac{\max_{v}R_{\mathrm{FH}}(\epsilon)}{R_{\mathrm{FD}}(\epsilon)} &\geq& \lim_{\gamma \to \infty}
\frac{\max_{v}R^{(3)}_{\mathrm{FH,lb}}(\epsilon)}{R_{\mathrm{FD}}(\epsilon)} \notag\\
&=& \frac{n_{\mathrm{des}} \left( 1-\frac{1}{n_{\max}} \right)^{n_{\max}-1}}{n_{\max}} \notag\\
&\stackrel{(a)}{>}& \frac{n_{\mathrm{des}}}{e n_{\max}},
\end{eqnarray}
where $(a)$ follows from the fact that $\left( 1-\frac{1}{n_{\max}} \right)^{n_{\max}-1} > \frac{1}{e}$ for any $n_{\max} \geq 1$. The above equation implies that $n_{\mathrm{des}} \geq e n_{\max}$ is a sufficient condition such that FH outperforms  FD in the high SNR regime.

\textit{Example 2}- Here, we consider a practical scenario with $n_{\mathrm{des}}=n_{\max}$ licensed users sharing $u=n_{\max}$ sub-bands. Each user can be active with some probability $p$ independently of other users. It is assumed that $n_{\max} \gg 1$ and $p \ll 1$ such that $\lambda\triangleq pn_{\max}$ is a constant. Let us assume that the $i^{th}$ user is active. We are looking for a sufficient condition to guarantee a better performance in terms of the $\epsilon$-outage capacity for this user in the FH scenario compared to FD. The number of users, other than the $i^{th}$ user, which are simultaneously active together with the $i^{th}$ user is a Binomial random variable with parameters $(n_{\max}-1,p)$. This random variable, denoted by $\widetilde{N}$, can be well approximated\footnote{We note that $\widetilde{N}=N-1$.} by a Poisson random variable with parameter $(n_{\max}-1)p \approx \lambda$. As $n_{\max}=n_{\mathrm{des}}$, designing the FH scenario based on this value is highly inefficient. In fact, for some $n^* \geq 1$ where $u$ is divisible by $n^*$, the outage probability in the FH system can be written as
\begin{eqnarray}
\label{hfd}
 \Pr \{\mathcal{O}_{i,\mathrm{FH}} (R)\} &=& \Pr \{\mathcal{O}_{i,\mathrm{FH}} (R) | N \leq n^*\} \Pr \{N \leq n^*\} + \Pr \{\mathcal{O}_{i,\mathrm{FH}} (R) | N > n^*\} \Pr \{N > n^*\} \notag\\
&\leq& \Pr \{\mathcal{O}_{i,\mathrm{FH}} (R) | N \leq n^*\} + \Pr \{N > n^*\} \notag\\
&\leq& \Pr \{\mathcal{O}_{i,\mathrm{FH}} (R) | N = n^*\} + \Pr \{N > n^*\}.
\end{eqnarray}
The last line follows from the fact that for a fixed $R$ the outage probability is an increasing function of the number of active users. Choosing $v=\frac{u}{n^{*}}$ in (\ref{kuyyywh10}) and selecting the transmission rate as
\begin{eqnarray} \label{kuyyywh12}
 R^{*}= \frac{u}{n^{*}} \log \left( 1- \kappa_{\frac{u}{n^{*}}} \gamma^{\left(1-\frac{1}{n^{*}}\right)^{n^{*}-1}} \ln \left( 1-\frac{\epsilon}{2}\right) \right)\notag\\
\end{eqnarray}
make $\Pr \{\mathcal{O}_{i,\mathrm{FH}} (R^{*}) | N = n^*\} \leq \frac{\epsilon}{2}$ where $\kappa_{\frac{u}{n^{*}}}=\frac{2^{-\mathscr{H}\left(\frac{u}{n^*},n^{*}\right)}n^{*}}{u}\left(\frac{n^*(n^{*}-1)}{u}\right)^{\left(1-\frac{1}{n^*}\right)^{n^{*}-1}-1}$. Furthermore, we select $n^*$ such that
\begin{eqnarray}
\Pr \{N > n^*\} \leq \frac{\epsilon}{2}.
\end{eqnarray}
Noting that
\begin{eqnarray}
\Pr \{N > n^*\} &=& \Pr \{\widetilde{N} \geq n^*\} \notag\\
&\approx& \sum_{n=n^*}^{\infty} \frac{e^{-\lambda} \lambda^n}{n!},
\end{eqnarray}
it is shown in appendix F that for small enough $\epsilon$, selecting $n^* = -\lambda \ln\epsilon$ guarantees $\Pr \{N > n^*\} \leq \frac{\epsilon}{2}$. Hence, using (\ref{kuyyywh12}),
\begin{eqnarray}
\label{poxb}
 R_{\mathrm{FH}} (\epsilon) &\stackrel{(a)}{\geq}&R^{*}\notag\\
 &\geq& \frac{u}{n^{*}} \log \left( - \kappa_{\frac{u}{n^{*}}} \gamma^{\left(1-\frac{1}{n^{*}}\right)^{n^{*}-1}} \ln \left( 1-\frac{\epsilon}{2}\right) \right)\notag\\
 &=&\frac{u}{n^{*}}\left(1-\frac{1}{n^{*}}\right)^{n^{*}-1}\log\gamma+\frac{u}{n^{*}}\log\left(-\kappa_{\frac{u}{n^{*}}}\ln\left(1-\frac{\epsilon}{2}\right)\right)\notag\\
 &\approx&\frac{u}{n^{*}}\left(1-\frac{1}{n^{*}}\right)^{n^{*}-1}\log\gamma\notag\\
 &\geq&  -\frac{u}{e \lambda \ln\epsilon} \log \gamma,
\end{eqnarray}
as
$\gamma \to \infty$. In (\ref{poxb}), $(a)$ follows by the fact that setting the transmission rate at $R^{*}$, we get $ \Pr \{\mathcal{O}_{i,\mathrm{FH}} (R^{*})\}\leq \epsilon$. This can be easily seen by (\ref{hfd}) for the particular choice of $n^{*}=-\lambda\ln\epsilon$.

In the FD scenario, since each active user is only allowed to utilize one frequency band, we have
\begin{eqnarray}
 R_{\mathrm{FD}} (\epsilon) \approx \frac{u}{n_{\mathrm{des}}} \log \gamma.
\end{eqnarray}
The above equations imply that as long as $n_{\mathrm{des}} > -e \lambda \ln\epsilon$, FH is superiour to FD. This condition can be alternatively written as $p < -\frac{1}{e \ln\epsilon}$. For example, if $u=n_{\mathrm{des}}=100$ and $\epsilon = 0.01$, the above condition is satisfied for $p < 0.08$. $\square$

\subsection{Numerical Results}
  In this section, we consider different examples to demonstrate cases where the FH scenario outperforms FD, i.e., $\max_{v}R_{\mathrm{FH}}(\epsilon)>R_{\mathrm{FD}}(\epsilon)$ for given parameters $n_{\mathrm{des}}, \{q_{n}\}_{n=1}^{n_{\max}},u, \epsilon$ and $\gamma$. We have no exact expression for $R_{\mathrm{FH}}(\epsilon)$. However, we have developed the following set of lower bounds on this quantity,
  \begin{equation}
  R_{\mathrm{FH}}(\epsilon)\geq R^{(1)}_{\mathrm{FH,lb}}(\epsilon)\geq R^{(2)}_{\mathrm{FH,lb}}(\epsilon)\geq R^{(3)}_{\mathrm{FH,lb}}(\epsilon),
  \end{equation}
  which provide us with sufficient conditions to observe supremacy of FH over FD.
  $R^{(1)}_{\mathrm{FH,lb}}(\epsilon)$ is the best lower bound, however, its computation involves multiple integrals of up to order $n_{\max}-1$ for any $n_{\max}$. Computing $R^{(2)}_{\mathrm{FH,lb}}(\epsilon)$ only involves a single integral for all values of $n_{\max}$, whereas computation of $R^{(3)}_{\mathrm{FH,lb}}(\epsilon)$ involves no integration. In the following example, we assume $n_{\max}\leq 4$. This enables us to use our best lower bound $ R^{(1)}_{\mathrm{FH,lb}}(\epsilon)$ as we are able to manipulate the double and triple integrals. In all numerical results, we also include the plots of the FBS scheme for the sake of comparison.

   \textit{Example 3}-
    Let $u=n_{\mathrm{des}}=20$, $n_{\max}=3$, $(q_{1},q_{2},q_{3})=(0.5, 0.3,0.2)$ and $\gamma=20\mathrm{dB}$. Fig. \ref{fig7} depicts $\max_{v}R^{(1)}_{\mathrm{FH,lb}}(\epsilon)$, $R_{\mathrm{FBS}}(\epsilon)$ and $R_{\mathrm{FD}}(\epsilon)$ as a function of $\epsilon$. It is seen that for $\epsilon>0.05$ the FH scenario offers a better outage capacity compared to the FD scheme. Also, for $\epsilon>0.16$, the FH scheme converges to the FBS scenario meaning that no advantage is observed by hopping over different sub-bands. Increasing the SNR to $30\mathrm{dB}$, fig. \ref{fig8} illustrates the complete dominance of FH for $\epsilon\geq 0.02$ over FBS and FD. $\square$
\begin{figure}[h!b!t]
  \centering
  \includegraphics[scale=.6] {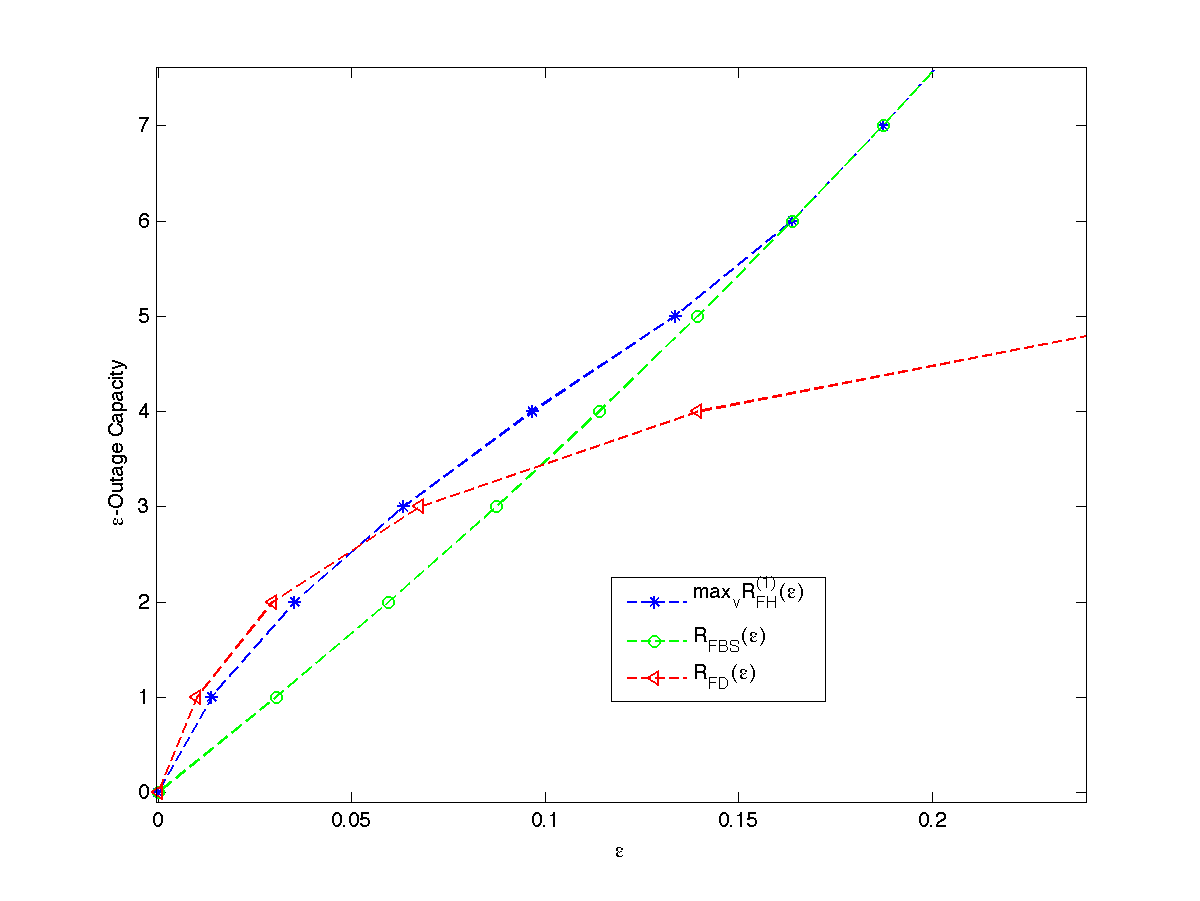}
  \caption{Comparison of FH and FD for $u=n_{\mathrm{des}}=20$, $n_{\max}=3$, $\gamma=20\mathrm{dB}$ and $(q_{1},q_{2},q_{3})=(0.5, 0.3,0.2)$.  }
  \label{fig7}
 \end{figure}
  \begin{figure}[h!b!t]
 \centering
  \includegraphics[scale=.6] {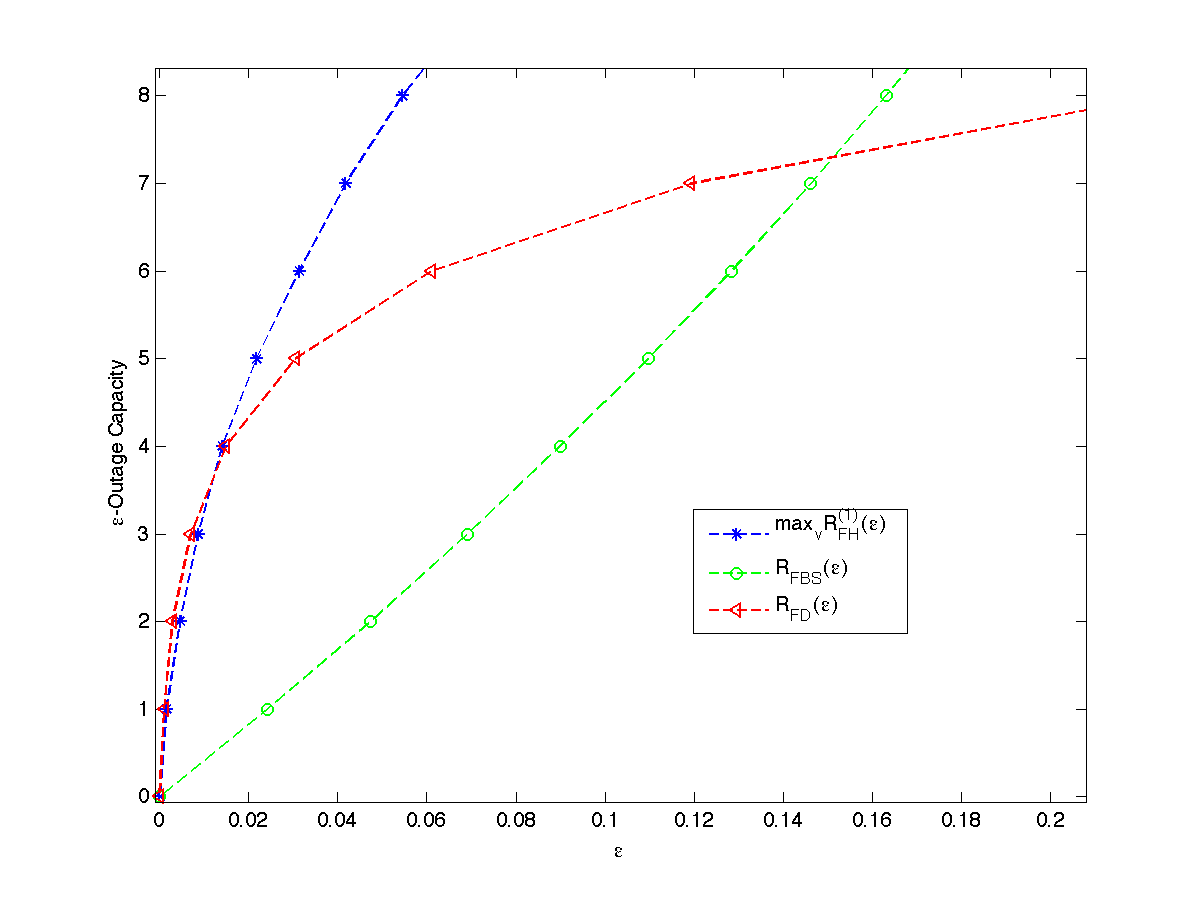}
  \caption{Comparison of FH and FD for $u=n_{\mathrm{des}}=20$, $n_{\max}=3$, $\gamma=30\mathrm{dB}$ and $(q_{1},q_{2},q_{3})=(0.5, 0.3,0.2)$.  }
  \label{fig8}
 \end{figure}

  \section{Conclusion}
  In this paper, we considered a decentralized  wireless communication network with a fixed number $u$ of
frequency sub-bands to be shared among $N$
transmitter-receiver pairs. It is assumed that the number of users $N$ is a random variable with a given distribution and the channel gains are quasi-static Rayleigh fading. The transmitters are assumed to be oblivious to the number of active users in the network as well as the channel gains. Moreover, the users are unaware of each other's codebooks and hence, no multiuser detection is possible.  We considered  the randomized Frequency Hopping scheme in which each transmitter randomly hops over $v$ out of $u$ sub-bands from transmission to transmission.  Developing a new upper bound on the differential entropy of a mixed Gaussian random vector and via entropy power inequality, we offered three lower bounds on the $\epsilon$-outage capacity for each user in the proposed scheme. Asymptotic analysis is presented in terms of  SNR and outage threshold. In the asymptotically small $\epsilon$ regime, we observed that the maximum outage capacity is obtained for either $v=1$ or $v=u$; in the asymptotically small SNR regime, we demonstrated that for all values of $v$ the system achieves the optimal performance; for asymptotically large SNR, it is shown that $v_{\mathrm{opt}}=\left\lceil \frac{u}{n_{\max}}\right\rceil$, where $n_{\max}$ is the maximum number of concurrently active users in the network. We compared the outage capacity of the underlying FH scheme with that of the FD scenario for various setups in terms of  distributions on the number of active users, SNR and $\epsilon$ and showed that FH outperforms FD in many cases.

 \section*{Appendix A; Proof of Lemma 1}
 Let us consider a $t\times 1$ complex vector mixed Gaussian distribution $p_{\vec{\Theta}}(\vec{\theta})$ with different covariance matrices $\{C_{l}\}_{l=1}^{L}$ and associated probabilities $\{p_{l}\}_{l=1}^{L}$ given by
\begin{equation}
\label{gb}
p_{\vec{\Theta}}(\vec{\theta})=\sum_{l=1}^{L}p_{l}g_{t}(\vec{\theta},C_{l}),\end{equation}
where $g_{t}(\vec{\theta},C_{l})=\frac{1}{\pi^{t}\det C_{l}}\exp\left(-\vec{\theta}^{H}C_{l}^{-1}\vec{\theta}\right)$.
 Hence,
\begin{equation} \label{pteta}
\int p_{\vec{\Theta}}(\vec{\theta})\log p_{\vec{\Theta}}(\vec{\theta})d\vec{\theta}=\sum_{l=1}^{L}J_{l}
\end{equation}
where $J_{l}=p_{l}\int g_{t}(\vec{\theta},C_{l})\log p_{\vec{\Theta}}(\vec{\theta})d\vec{\theta}$ for $1\leq l\leq L$.
To find a proper lower bound on each $J_{l}$ in this expression, we proceed as follows. We know that $C_{l}=\varrho_{l}^{2}I_{t}$ and $\varrho_{1}^{2}<\varrho_{2}^{2}<\cdots<\varrho_{L}^{2}$. We have
\begin{equation}
\label{sx}
J_{l}= p_{l}\int g_{t}(\vec{\theta},C_{l})\log\left(\sum_{m=1}^{L}p_{m}g_{t}(\vec{\theta},C_{m})\right)d\vec{\theta}.\end{equation}
On the other hand,
\begin{equation}
\label{ssx}
\log\left(\sum_{m=1}^{L}p_{m}g_{t}(\vec{\theta},C_{m})\right)
=\log \left(p_{l}g_{t}(\vec{\theta},C_{l})\right)+\log\left(1+\sum_{m=1}^{l-1}\frac{p_{m}}{p_{l}}\frac{g_{t}(\vec{\theta},C_{m})}{g_{t}(\vec{\theta},C_{l})}+\sum_{m=l+1}^{L}\frac{p_{m}}{p_{l}}\frac{g_{t}(\vec{\theta},C_{m})}{g_{t}(\vec{\theta},C_{l})}\right).\end{equation}
However, the term $\sum_{m=l+1}^{L}\frac{p_{m}}{p_{l}}\frac{g_{t}(\vec{\theta},C_{m})}{g_{t}(\vec{\theta},C_{l})}=\sum_{m=l+1}^{L}\frac{p_{m}}{p_{l}}\frac{\varrho_{l}^{2t}}{\varrho_{m}^{2t}}\exp-\big((\frac{1}{\varrho_{m}^{2}}-\frac{1}{\varrho_{l}^{2}})\vec{\theta}^{H}\vec{\theta}\big)$ is always greater than $\sum_{m=l+1}^{L}\frac{p_{m}}{p_{l}}\frac{\varrho_{l}^{2t}}{\varrho_{m}^{2t}}$. Hence,
\begin{equation}
\label{nb}
\log\left(\sum_{m=1}^{l}p_{m}g_{t}(\vec{\theta},C_{m})\right)\geq \log \left(p_{l}g_{t}(\vec{\theta},C_{l})\right)+\log\left(1+\sum_{m=l+1}^{L}\frac{p_{m}}{p_{l}}\frac{\varrho_{l}^{2t}}{\varrho_{m}^{2t}}+\sum_{m=1}^{l-1}\frac{p_{m}}{p_{l}}\frac{g_{t}(\vec{\theta},C_{m})}{g_{t}(\vec{\theta},C_{l})}\right).\end{equation}

On the other hand, the term $\sum_{m=1}^{l-1}\frac{p_{m}}{p_{l}}\frac{g_{t}(\vec{\theta},C_{m})}{g_{t}(\vec{\theta},C_{l})}=\sum_{m=1}^{l-1}\frac{p_{m}}{p_{l}}\frac{\varrho_{l}^{2t}}{\varrho_{m}^{2t}}\exp-\big((\frac{1}{\varrho_{m}^{2}}-\frac{1}{\varrho_{l}^{2}})\vec{\theta}^{H}\vec{\theta}\big)$ is always less than $\sum_{m=1}^{l-1}\frac{p_{m}}{p_{l}}\frac{\varrho_{l}^{2t}}{\varrho_{m}^{2t}}$. Now, we use the following inequality\footnote{One may verify this using Jensen's inequality and concavity of the $\log(.)$ function.}, which is valid for any $b>0$ and $0\leq x\leq a$,
\begin{equation}
\log(1+b+x)\geq (1-\frac{x}{a})\log(1+b)+\frac{x}{a}\log(1+a+b).
\end{equation}
Utilizing this in the expression on the right-hand side of (\ref{nb}), we get:
\begin{equation}
\log\left(\sum_{m=1}^{l}p_{m}g_{t}(\vec{\theta},C_{m})\right)\geq \left(1-\frac{1}{\nu_{l}}\sum_{m=1}^{l-1}\frac{p_{m}}{p_{l}}\frac{g_{t}(\vec{\theta},C_{m})}{g_{t}(\vec{\theta},C_{l})}\right)\log(1+\mu_{l})+\frac{1}{\nu_{l}}\sum_{m=1}^{l-1}\frac{p_{m}}{p_{l}}\frac{g_{t}(\vec{\theta},C_{m})}{g_{t}(\vec{\theta},C_{l})}\log(1+\nu_{l}+\mu_{l})\end{equation}
where $\mu_{l}=\sum_{m=l+1}^{L}\frac{p_{m}}{p_{l}}\frac{\varrho_{l}^{2t}}{\varrho_{m}^{2t}}$ and $\nu_{l}=\sum_{m=1}^{l-1}\frac{p_{m}}{p_{l}}\frac{\varrho_{l}^{2t}}{\varrho_{m}^{2t}}$.
Using this in (\ref{sx}) yields:
\begin{equation}
J_{l}\geq p_{l}\int g_{t}(\vec{\theta},C_{l})\log(p_{l}g_{t}(\vec{\theta},C_{l}))d\vec{\theta}+\left(p_{l}-\frac{\sum_{m=1}^{l-1}p_{m}}{\nu_{l}}\right)\log(1+\mu_{l})+\frac{\sum_{m=1}^{l-1}p_{m}}{\nu_{l}}\log(1+\nu_{l}+\mu_{l}).
\end{equation}
The first term on the right-hand side can be calculated as
\begin{eqnarray}
\label{gbb}
 p_{l}\int g_{t}(\vec{\theta},C_{l})\log\big(p_{l}g_{t}(\vec{\theta},C_{l})\big)d\vec{\theta}&=& \left(p_{l}\log p_{l}\right)\int g_{t}(\vec{\theta},C_{l})d\vec{\theta}+p_{l}\int g_{t}(\vec{\theta},C_{l})\log g_{t}(\vec{\theta},C_{l})d\vec{\theta}\notag\\
&=& p_{l}\log p_{l}+p_{l}\int g_{t}(\vec{\theta},C_{l})\log g_{t}(\vec{\theta},C_{l})d\vec{\theta}\notag\\
&\stackrel{(a)}{=}&p_{l}\log p_{l}-p_{l}\log\left((\pi e)^{t}\det C_{l}\right)\notag\\
&=&p_{l}\log p_{l}-tp_{l}\log(\pi e\varrho_{l}^{2})
\end{eqnarray}
where in $(a)$ we have used the fact that the differential entropy of a $t\times 1$ complex Gaussian vector with covariance matrix $C_{l}$ is $\log\left((\pi e)^{t}\det C_{l}\right)$.
Thus,
\begin{eqnarray}
\mathrm{h}(\vec{\Theta}) &=& -\int p_{\vec{\Theta}}(\vec{\theta})\log p_{\vec{\Theta}}(\vec{\theta})d\vec{\theta}\notag\\&=&-\sum_{l=1}^{L}J_{l} \notag\\
&\leq& -\sum_{l=1}^{L}p_{l}\log p_{l}+t\sum_{l=1}^{L}p_{l}\log\left(\pi e\varrho_{l}^{2}\right) \notag\\
&&-\sum_{l=1}^{L}\left(\left(p_{l}-\frac{\sum_{m=1}^{l-1}p_{m}}{\nu_{l}}\right)\log(1+\mu_{l})+\frac{\sum_{m=1}^{l-1}p_{m}}{\nu_{l}}\log(1+\nu_{l}+\mu_{l})\right)\notag\\&=&t\sum_{l=1}^{L}p_{l}\log\left(\pi e\varrho_{l}^{2}\right)-\sum_{l=1}^{L}p_{l}\log p_{l} \notag\\
&&-\sum_{l=1}^{L}\left(\left(p_{l}-\frac{\sum_{m=1}^{l-1}p_{m}}{\nu_{l}}\right)\log(1+\mu_{l})+\frac{\sum_{m=1}^{l-1}p_{m}}{\nu_{l}}\log(1+\nu_{l}+\mu_{l})\right)\notag\\&=&t\sum_{l=1}^{L}p_{l}\log\left(\pi e\varrho_{l}^{2}\right)-\sum_{l=1}^{L}p_{l}\log p_{l}\notag\\
&&-\sum_{l=1}^{L}\left(\left(p_{l}-\frac{\sum_{m=1}^{l-1}p_{m}}{\nu_{l}}\right)\log(1+\mu_{l})+\frac{\sum_{m=1}^{l-1}p_{m}}{\nu_{l}}\log(1+\nu_{l}+\mu_{l})\right).\notag\\\end{eqnarray}
   Briefly,
 \begin{equation}
 \mathrm{h}(\vec{Z})\leq t\sum_{l=1}^{L}p_{l}\log(\pi e\varrho_{l}^{2})+\mathscr{H}-\mathscr{G}''\end{equation}
 where
 \begin{equation}
\label{ }
\mathscr{H}=-\sum_{l=1}^{L}p_{l}\log p_{l}
\end{equation}
and
 \begin{equation}
\mathscr{G}''=\sum_{l=1}^{L}\left(\left(p_{l}-\frac{\sum_{m=1}^{l-1}p_{m}}{\nu_{l}}\right)\log(1+\mu_{l})+\frac{\sum_{m=1}^{l-1}p_{m}}{\nu_{l}}\log(1+\nu_{l}+\mu_{l})\right)\notag\\. \end{equation}
$\mathscr{G}''$ is a complicated function of $\{\varrho_{l}\}_{l=1}^{L}$. To simplify it, one may notice that $\mathscr{G}''$ is an increasing function of $\mu_{l}$. Hence, using $\mu_{l}\geq 0$, we get a lower bound on $\mathscr{G}''$, namely $\mathscr{G}'$ given by
\begin{equation}
\mathscr{G}'=\sum_{l=2}^{L}\frac{\log(1+\nu_{l})}{\nu_{l}}\sum_{m=1}^{l-1}p_{m}.\end{equation}
 On the other hand, using the fact that $\frac{\log(1+x)}{x}$ is a decreasing function of $x$, one may obtain a lower bound on $\mathscr{G}'$ by finding an upper bound on $\nu_{l}$ for each $l$. One option is $\nu_{l}\leq \frac{\varrho_{L}^{2t}}{\varrho_{1}^{2t}}\frac{\sum_{m=1}^{l-1}p_{m}}{p_{l}}$. Thus, we come up with the following lower bound on $\mathscr{G}'$
 \begin{equation}
 \mathscr{G}'\geq \mathscr{G}\triangleq\frac{\varrho_{1}^{2t}}{\varrho_{L}^{2t}}\sum_{l=2}^{L}p_{l}\log\left(1+\frac{\varrho_{L}^{2t}}{\varrho_{1}^{2t}}\frac{\sum_{m=1}^{l-1}p_{m}}{p_{l}}\right).
 \end{equation}

\section*{Appendix B; Computation of $R^{(1)}_{\mathrm{FH,lb}}(\epsilon)$}
By (\ref{kood}),
  \begin{equation}\Pr\left\{\mathscr{R}^{(1)}_{i,\mathrm{lb}}(\vec{h}_{i})<R\right\}=\Pr\left\{v\log\left(\frac{2^{-\mathscr{H}(v,N)}2^{\alpha_{N}(\mathcal{I}_{N-1,1};\frac{\gamma}{v},\frac{v}{u})}\vert h_{i,i}\vert^{2}\gamma}{v\prod_{m=1}^{N-1}\prod_{m'=1}^{{N-1\choose m}}\left(\frac{\gamma}{v}\mathcal{I}_{m,m'}+1\right)^{\beta_{m,N}\left(\frac{v}{u}\right)}}+1\right)<R\right\}\end{equation}
  where $\beta_{m,N}\left(\frac{v}{u}\right)=(\frac{v}{u})^{m}\left(1-\frac{v}{u}\right)^{N-1-m}$ and for each $m$, $\{\mathcal{I}_{m,m'}\}_{m'=1}^{{N-1\choose m}}$ consists of all possible summations of $m$ elements in the set $\{\vert h_{j,i}\vert^{2}\}_{j=1,j\neq i}^{N}$, such that $\prod_{m=1}^{N-1}\prod_{m'=1}^{{N-1\choose m}}(\frac{\gamma}{v}\mathcal{I}_{m,m'}+1)^{\beta_{m,N}\left(\frac{v}{u}\right)}=\prod_{l=1}^{L_{i}}(c_{i,l}\gamma+1)^{a_{i,l}}$. We have also substituted $\mathscr{G}_i$ by $\mathscr{G}_{i,\mathrm{lb}} (v,N) =  \alpha_{N}(\mathcal{I}_{N-1,1};\frac{\gamma}{v},\frac{v}{u})$. Since $N$ itself is a random variable, the outage probability can be written as
  \begin{equation}
  \Pr\{\mathscr{R}^{(1)}_{i,\mathrm{lb}}(\vec{h}_{i})<R\}=q_{1}\xi_{1}+\sum_{n=2}^{\infty}q_{n}\xi_{n},
    \end{equation}
    where
        \begin{eqnarray}
    \xi_{1}&=&\Pr\{\mathscr{R}^{(1)}_{i,\mathrm{lb}}(\vec{h}_{i})<R|N=1\}\notag\\
&=&\Pr\left\{v\log\left(1+\frac{|h_{i,i}|^2\gamma}{v}\right)<R\right\} \notag\\ &=& 1-\exp\left(\frac{\left(1-2^{\frac{R}{v}}\right)v}{\gamma}\right).
        \end{eqnarray}
      Denoting the collection of random variables $\{\mathcal{I}_{m,m'}\}_{\substack{1\leq m\leq n-1\\1\leq m'\leq {n-1\choose m}}}$ by $\mathscr{I}_n$,
        \begin{eqnarray}
        \xi_{n}&=&\Pr\{\mathscr{R}^{(1)}_{i,\mathrm{lb}}(\vec{h}_{i})<R|N=n\}\notag\\
&=&\Pr\left\{v\log\left(\frac{2^{-\mathscr{H}(v,n)}2^{\alpha_{n}(\mathcal{I}_{n-1,1};\frac{\gamma}{v},\frac{v}{u})}\vert h_{i,i}\vert^{2}\gamma}{v\prod_{m=1}^{n-1}\prod_{m'=1}^{{n-1\choose m}}(\frac{\gamma}{v}\mathcal{I}_{m,m'}+1)^{\beta_{m,n}\left(\frac{v}{u}\right)}}+1\right)<R\right\}\notag\\
&=&\mathrm{E}\left\{\Pr\left\{v\log\left(\frac{2^{-\mathscr{H}(v,n)}2^{\alpha_{n}(\mathcal{I}_{n-1,1};\frac{\gamma}{v},\frac{v}{u})}\vert h_{i,i}\vert^{2}\gamma}{v\prod_{m=1}^{n-1}\prod_{m'=1}^{{n-1\choose m}}(\frac{\gamma}{v}\mathcal{I}_{m,m'}+1)^{\beta_{m,n}\left(\frac{v}{u}\right)}}+1\right)<R\Big|\mathscr{I}_n\right\}\right\}\notag\\
&=&\mathrm{E}\left\{1-\exp\left(\frac{2^{\mathscr{H}(v,n)}\left(1-2^{\frac{R}{v}}\right)v}{\gamma}2^{-\alpha_{n}(\mathcal{I}_{n-1,1};\frac{\gamma}{v},\frac{v}{u})}\prod_{m=1}^{n-1}\prod_{m'=1}^{{n-1\choose m}}\left(\frac{\gamma}{v}\mathcal{I}_{m,m'}+1\right)^{\beta_{m,n}\left(\frac{v}{u}\right)}\right)\right\}\notag\\
&\stackrel{(a)}{=}&1-\mathrm{E}\left\{\exp\left(\frac{2^{\mathscr{H}(v,n)}\left(1-2^{\frac{R}{v}}\right)v}{\gamma}2^{-\alpha_{n}(\mathcal{I}_{n-1,1};\frac{\gamma}{v},\frac{v}{u})}\prod_{m=1}^{n-1}\prod_{m'=1}^{{n-1\choose m}}\left(\frac{\gamma}{v}\mathcal{I}_{m,m'}+1\right)^{\beta_{m,n}\left(\frac{v}{u}\right)}\right)\right\}\notag\\
&=&1-\psi_{n}\left(\frac{2^{\mathscr{H}(v,n)}\left(1-2^{\frac{R}{v}}\right)v}{\gamma},\frac{\gamma}{v},\frac{v}{u}\right),     \end{eqnarray}
where $(a)$ follows from the fact that after conditioning on $\mathscr{I}_n$, the only random variable is $\vert h_{i,i}\vert^{2}$, which is exponentially distributed.
        As a result, using the definition of $\psi_n $,
        \begin{equation}
        R^{(1)}_{\mathrm{FH,lb}}(\epsilon)=\sup\left\{ R:q_{1}\exp\left(\frac{\left(1-2^{\frac{R}{v}}\right)v}{\gamma}\right)+\sum_{n=2}^{\infty}q_{n}\psi_{n}\left(\frac{2^{\mathscr{H}(v,n)}\left(1-2^{\frac{R}{v}}\right)v}{\gamma},\frac{\gamma}{v},\frac{v}{u}\right)>1-\epsilon\right\}.
        \end{equation}

        \section*{Appendix C; Computation of $R^{(2)}_{\mathrm{FH,lb}}(\epsilon)$}
        By (\ref{koodie}),
        \begin{equation}
        \mathscr{R}^{(2)}_{i,\mathrm{lb}} (\vec{h}_i)=v\log\left(\frac{2^{-\mathscr{H}(v,N)}2^{\mathscr{G}_{i,\mathrm{lb}}(v,N)}\vert h_{i,i}\vert^{2}\gamma}{v(\frac{\gamma}{v}\mathcal{J}_{N}+1)^{1-a(v,N)}}+1\right),               \end{equation}
        where
\begin{equation}
\label{chos1}
\mathcal{J}_{N}=\left\{\begin{array}{cc }
    \sum_{j=1,j\neq i}^{N}\vert h_{j,i}\vert^{2}  & N>1   \\
    0  &   N=1
\end{array}\right..
\end{equation}

        Thus, following the same lines as in appendix B, we have
        \begin{equation}
        \Pr\{\mathscr{R}^{(2)}_{i,\mathrm{lb}} (\vec{h}_i)<R\}=q_{1}\xi_{1}+\sum_{n=2}^{\infty}q_{n}\xi_{n}
                \end{equation}
        where
        \begin{equation}
        \xi_{1}=1-\exp\left(\frac{\left(1-2^{\frac{R}{v}}\right)v}{\gamma}\right),
        \end{equation}
  and
  \begin{equation}
  \label{poor}
  \xi_{n}=1-\mathrm{E}\left\{\exp\left(\frac{2^{\mathscr{H}(v,n)}\left(1-2^{\frac{R}{v}}\right)v}{\gamma}2^{-\alpha_{n}(\mathcal{J}_{n};\frac{\gamma}{v},\frac{v}{u})} \left(\frac{\gamma}{v}\mathcal{J}_{n}+1 \right)^{1-a(v,n)}\right)\right\},   \end{equation}
  for all $n\geq 2$. Since $2\mathcal{J}_{n}\sim\chi^{2}_{2(n-1)}$, we have $p_{\mathcal{J}_{n}}(z)=\frac{1}{(n-2)!}z^{n-2}\exp\left(-z\right)$.
  Therefore, using the definition of $\phi_n $, (\ref{poor}) can be expressed as
  \begin{equation}
  \xi_{n}=1-\phi_{n}\left( \frac{2^{\mathscr{H}(v,n)}\left(1-2^{\frac{R}{v}}\right)v}{\gamma},\frac{\gamma}{v},1-a(v,n),\frac{v}{u}\right).
    \end{equation}
    Hence, $R^{(2)}_{\mathrm{FH,lb}}(\epsilon)$ is given by
    \begin{equation}
    R^{(2)}_{\mathrm{FH,lb}}(\epsilon)= \sup\left\{R: q_{1}\exp(\frac{\left(1-2^{\frac{R}{v}}\right)v}{\gamma})+\sum_{n=2}^{\infty}q_{n}\phi_{n}\left(\frac{2^{\mathscr{H}(v,n)}\left(1-2^{\frac{R}{v}}\right)v}{\gamma},\frac{\gamma}{v},1-a(v,n),\frac{v}{u}\right)>1-\epsilon\right\}.
    \end{equation}
    \section*{Appendix D; Proof of Corollary 1}
As $b_{1}<0$ and $0<2^{-\alpha_{n}(\theta;b_{2},c_{2})}<1$ for all $\theta>0$, we have
  \begin{equation}
  \label{nooo}
  \phi_{n}(b_{1},b_{2},c_{1},c_{2})\geq \frac{1}{(n-2)!}\int_{0}^{\infty}\theta^{n-2}\exp\big(b_{1}(b_{2}\theta+1)^{c_{1}}-\theta\big)d\theta.  \end{equation}
  One may easily check that $\frac{\theta^{n-2}}{(n-2)!}\exp(-\theta)\mathbb{1}(\theta>0)$ is a PDF for some nonnegative random variable $\Theta$.  Thus, (\ref{nooo}) can be written as
  \begin{equation}
  \phi_{n}(b_{1},b_{2},c_{1},c_{2})\geq \mathrm{E}\left\{\exp\left(b_{1}(b_{2}\Theta+1)^{c_{1}}\right)\right\}.
  \end{equation}
  However, as $b_{1}<0$ and $0<c_{1}<1$, the function $\exp\big(b_{1}(b_{2}\Theta+1)^{c_{1}}\big)$ is a convex function of $\Theta$. Hence, applying Jensen's inequality yields
  \begin{equation}
  \label{peh}
  \phi_{n}(b_{1},b_{2},c_{1},c_{2})\geq \exp\big(b_{1}(b_{2}\mathrm{E}\{\Theta\}+1)^{c_{1}}\big)=\exp\Big(b_{1}\big((n-1)b_{2}+1\big)^{c_{1}}\Big), \end{equation}
  where we have used $\mathrm{E}\{\Theta\}=n-1$. Using (\ref{peh}) in (\ref{pqqq}) and noting that $b_1=\frac{\left(1-2^{\frac{R}{v}}\right)v}{\gamma}2^{\mathscr{H}(v,n)}$, $b_2=\frac{\gamma}{v}$, and $c_1=1-a(v,n)$, we get the desired lower bound.

\section*{Appendix E; Computation of $R_{\mathrm{FBS}}(\epsilon)$}
We first compute $\Pr\{\mathscr{R}_{i,\mathrm{FBS}}(\vec{h}_{i})<R\}$. By (\ref{tg}),
\begin{equation}
\Pr\{\mathscr{R}_{i,\mathrm{FBS}}(\vec{h}_{i})<R\}=\Pr\left\{u\log\left(1+\frac{\vert h_{i,i}\vert^{2}\gamma}{u(1+\frac{\gamma}{u}\mathcal{J}_{N})}\right)<R\right\},
\end{equation}
where $\mathcal{J}_N$ is defined in (\ref{chos1}).
Therefore,
\begin{eqnarray}
\Pr\{\mathscr{R}_{i,\mathrm{FBS}}(\vec{h}_{i})<R\}&=&\sum_{n=1}^{\infty}q_{n}\Pr\left\{\mathscr{R}_{i,\mathrm{FBS}}(\vec{h}_{i})<R\Big| N=n\right\}\notag\\
&=&q_{1}\Pr\left\{u\log\left(1+\frac{\vert h_{i,i}\vert^{2}\gamma}{u}\right)<R\right\}\notag\\&&+\sum_{n=2}^{N}q_{n}\Pr\left\{u\log\left(1+\frac{\vert h_{i,i}\vert^{2}\gamma}{u(1+\frac{\gamma}{u}\mathcal{J}_{n})}\right)<R\right\}.
\end{eqnarray}
The first term can be computed easily as
\begin{equation}
\Pr\left\{v\log\left(1+\frac{\vert h_{i,i}\vert^{2}\gamma}{u}\right)<R\right\}=1-\exp\left(\frac{(1-2^{\frac{R}{u}})u}{\gamma}\right).
\end{equation}
For any $n\geq 2$, one can write
\begin{equation}
\label{py}
\Pr\left\{u\log\left(1+\frac{\vert h_{i,i}\vert^{2}\gamma}{u(1+\frac{\gamma}{u}\mathcal{J}_{n})}\right)<R\right\}=\mathrm{E}\left\{\Pr\left\{u\log\left(1+\frac{\vert h_{i,i}\vert^{2}\gamma}{u(1+\frac{\gamma}{u}\mathcal{J}_{n})}\right)<R\Big| \mathcal{J}_{n}\right\}\right\}
\end{equation}
Since $\vert h_{i,i}\vert^{2}$ is an exponential random variable with parameter one,
\begin{equation}
\Pr\left\{u\log\left(1+\frac{\vert h_{i,i}\vert^{2}\gamma}{u(1+\frac{\gamma}{u}\mathcal{J}_{n})}\right)<R\Big |\mathcal{J}_{n}\right\}=1-\exp\left(\frac{(1-2^{\frac{R}{u}})(1+\frac{\gamma}{u}\mathcal{J}_{n})u}{\gamma}\right).\end{equation}
Replacing this in (\ref{py}) yields
\begin{eqnarray}
\Pr\left\{u\log\left(1+\frac{\vert h_{i,i}\vert^{2}\gamma}{u(1+\frac{\gamma}{u}\mathcal{J}_{n})}\right)<R\right\}&=&\mathrm{E}\left\{1-\exp\left(\frac{(1-2^{\frac{R}{u}})(1+\frac{\gamma}{u}\mathcal{J}_{n})u}{\gamma}\right)\right\}\notag\\&=&1-\exp\left(\frac{(1-2^{\frac{R}{u}})u}{\gamma}\right)\mathrm{E}\left\{\exp\big((1-2^{\frac{R}{u}})\mathcal{J}_{n}\big)\right\}\notag\\&=&1-\exp\left(\frac{(1-2^{\frac{R}{u}})u}{\gamma}\right)2^{-\frac{(n-1)R}{u}}\end{eqnarray}
where we have used the fact that $\mathrm{E}\{\exp(t\mathcal{J}_{n})\}=\frac{1}{(1-t)^{n-1}}$ as $2\mathcal{J}_{n}\sim\chi^{2}_{2(n-1)} $. Thus, $R_{\mathrm{FBS}}(\epsilon)$ is given by
\begin{equation}
R_{\mathrm{FBS}}(\epsilon)=\sup\left\{R: \exp\left(\frac{(1-2^{\frac{R}{u}})u}{\gamma}\right)\sum_{n=1}^{\infty}q_{n}2^{-\frac{(n-1)R}{u}}>1-\epsilon\right\}.
\end{equation}

\section*{Appendix F}
Setting $n^* = -\lambda \ln\epsilon$, we have
\begin{eqnarray}
 \Pr \{\widetilde{N} \geq n^*\} &\approx& \sum_{n=n^*}^{\infty} \frac{e^{-\lambda} \lambda^n}{n!} \notag\\
&=& e^{-\lambda} \frac{\lambda^{n^*}}{n^*!} \left( 1+ \frac{\lambda}{n^*+1} + \frac{\lambda^2}{(n^*+1)(n^*+2)} + \cdots\right) \notag\\
&\leq&  e^{-\lambda} \frac{\lambda^{n^*}}{n^*!} \left[ \sum_{j=0}^{\infty} \left( \frac{\lambda}{n^*}\right)^j\right] \notag\\
&=& e^{-\lambda} \frac{\lambda^{n^*}}{n^*!} \frac{1}{1-\frac{\lambda}{n^*}} \notag\\
&\stackrel{(a)}{\approx}& \frac{e^{-\lambda}}{\sqrt{2 \pi n^*}} \left(\frac{\lambda e}{n^*}\right)^{n^*}  \frac{1}{1-\frac{\lambda}{n^*}} \notag\\
&=&  \frac{e^{-\lambda}}{\sqrt{2 \pi n^*}} \left(\frac{e}{-\ln\epsilon}\right)^{n^*} \frac{1}{1+\frac{1}{\ln\epsilon}}
\end{eqnarray}
where $(a)$ follows by stirling approximation for $n^{*}!$.
The term $\left(\frac{e}{-\ln\epsilon}\right)^{n^*}$ can be written as
\begin{eqnarray}
\left(\frac{e}{-\ln\epsilon}\right)^{n^*} &=& e^{-n^*\left(\ln (-\ln\epsilon)-1 \right)} \notag\\
&=& e^{\lambda \left( \ln( -\ln\epsilon)-1\right)\ln\epsilon} \notag\\
&\stackrel{(a)}{\leq}& e^{\ln\epsilon} \notag\\
&=& \epsilon.
\end{eqnarray}
where $(a)$ is valid if $\lambda(\ln(-\ln\epsilon)-1)\geq1$, which is the case for sufficiently small $\epsilon$. Combining this with the fact that $\frac{e^{-\lambda}}{\sqrt{2 \pi n^*}} \frac{1}{1+\frac{1}{\ln\epsilon}} < \frac{1}{2}$ gives the desired result.

    \bibliographystyle{IEEEbib}
      
\end{document}